\documentclass[conference]{IEEEtran}
\usepackage{cite}
\usepackage{amsmath,amssymb,amsfonts}
\usepackage{graphicx}
\usepackage{algorithmic}
\usepackage{textcomp}
\usepackage{hyperref}

\usepackage{cite}
\usepackage{times}
\usepackage{multicol}
\usepackage{caption}
\usepackage{subcaption}
\usepackage{bm}
\usepackage[linesnumbered,ruled,vlined]{algorithm2e}
\usepackage{pifont}   
\usepackage{orcidlink}
\usepackage{booktabs}
\usetikzlibrary{arrows.meta, positioning}

\newcommand{\Rr}{{\mathbb{R}}}

\newcommand{\Aa}{{\mathcal{A}}}

\def\dt{{\rm d}t}
\def\leq{\leqslant}
\def\geq{\geqslant}

\newtheorem{theorem}{Theorem}

\newtheorem{proposition}[theorem]{Proposition}
\newtheorem{theorem*}{Theorem}
\newtheorem{lemma*}[theorem]{Lemma}
\newtheorem{corollary*}[theorem]{Corollary}
\newtheorem{proposition*}[theorem]{Proposition}
\newtheorem{problem*}[theorem]{Problem}
\newtheorem{definition}[theorem]{Definition}
\newtheorem{hyp}{Assumption}
\newtheorem{problem}{Problem}
\newtheorem{remark}[theorem]{Remark}

\def\BibTeX{{\rm B\kern-.05em{\sc i\kern-.025em b}\kern-.08em
    T\kern-.1667em\lower.7ex\hbox{E}\kern-.125emX}}
\begin{document}
\title{Game-Theoretic Coordination for Time-Critical Missions of UAV Systems}
\author{Mikayel Aramyan \orcidlink{0009-0005-4496-7250}, Anna Manucharyan \orcidlink{0009-0003-5275-8032}, Lusine Poghosyan \orcidlink{0000-0002-9535-8966}, Tigran Bakaryan \orcidlink{0000-0002-6633-4180} and Naira Hovakimyan \orcidlink{0000-0003-3850-1073}, \IEEEmembership{Fellow, IEEE}
\thanks{The research is supported by AFOSR grant \#FA9550-21-1-0411 and NASA grant \#80NSSC22M0070, and in part by  the Higher Education and Science Committee of the MESCS RA under Research Project No.~24IRF-1A001 (supporting A.M., L.P., and T.B.).}
\thanks{Mikayel Aramyan and Naira Hovakimyan Authors are with the Department of Mechanical Science \& Engineering, University of Illinois Urbana-Champaign, Urbana, IL 61801, USA (e-mail: mikayel2@illinois.edu;  nhovakim@illinois.edu).}
\thanks{Anna Manucharyan is with the Center For Scientific Innovation and Education, Yerevan, Armenia, and Akian College of Science and Engineering, American University of Armenia, Yerevan, Armenia (e-mail: anna.manucharyan@csie.am).}
\thanks{ Lusine Poghosyan and Tigran Bakaryan are with the Center for Scientific Innovation and Education, and the Institute of Mathematics of NAS RA, and the Yerevan State University, Yerevan, Armenia (e-mail: lusine@instmath.sci.am; tigran.bakaryan@instmath.sci.am).}}

\maketitle

\begin{abstract}
Coordinated missions involving Unmanned Aerial Vehicles (UAVs) in dynamic environments pose significant challenges in maintaining both coordination and agility. In this paper, relying on the cooperative path following framework and using a game-theoretic formulation, we introduce a novel and scalable approach in which each UAV acts autonomously in different mission conditions. This formulation naturally accommodates heterogeneous and time-varying objectives across the system.  
In our setting, each UAV optimizes a cost function that incorporates temporal and mission-specific constraints. The optimization is performed within a one-dimensional domain, significantly reducing the computational cost and enabling real-time application to complex and dynamic scenarios. The framework is distributed in structure, enabling global, system-wide coordination (a Nash equilibrium) by using only local information.
For ideal systems, we prove the existence and  the Nash equilibrium exhibits exponential convergence. Furthermore,  we invoke model predictive control (MPC) for non-ideal scenarios. In particular, we propose a discrete-time optimization approach that tackles path-following errors and communication failures, ensuring reliable and agile performance in dynamic and uncertain environments. Simulation results demonstrate the effectiveness and agility of the approach in ensuring successful mission execution across diverse realistic scenarios. 
\end{abstract}


\section{Introduction}
Recent advancements in technology have significantly expanded the capabilities of Unmanned Aerial Vehicle (UAV) systems, enabling them to perform a wide range of tasks \cite{mohsan2023unmanned, UAVApp}. Particularly, cooperative UAV systems have been extensively deployed across diverse domains, including civilian applications such as surveillance, environmental monitoring, and air traffic management, as well as military operations like swarm-based attack-defense scenarios. Notable examples include cooperative forest fire monitoring and suppression \cite{pham2018distributed, sujit2007cooperative}, surveillance of multiple moving targets \cite{gu2018multiple}, to name a few. Despite their broad range of applications, cooperative UAV systems face operational limitations and challenges like battery endurance, payload carrying capability, flight autonomy, path planning, path following, and achieving reliable cooperation.

Effective cooperation among a system of UAVs is crucial for accomplishing complex tasks that exceed the capabilities of a single UAV. Applications such as simultaneous target tracking \cite{swarm-track}, formation flying \cite{swar-flock}, and large-scale area mapping and monitoring \cite{UAV-App} heavily rely on the coordinated efforts of multiple UAVs. Achieving reliable cooperation, however, is inherently challenging due to factors such as communication delays, limited computational resources, variations in UAV dynamics, disturbances, and unpredictable events. These challenges highlight the need for efficient, scalable and adaptive strategies to ensure system-wide collaboration in dynamic and uncertain environments. 
This paper addresses the challenge of time-critical cooperation by a novel game-theoretic method.

We develop a scalable distributed coordination framework for multi-UAV systems operating in realistic and dynamic environments. Since communication networks and disturbances are time-varying and uncertain, complete global information is generally unavailable. As a result, each UAV must retain local autonomy and adapt its objective online based on its own state and interactions with neighbors.
This naturally motivates a game-theoretic formulation, which accommodates heterogeneous objectives, such as situations where some UAVs prioritize collision avoidance while others focus on temporal coordination (see Section~\ref{subsec-collison-avoid}). Moreover, the proposed approach is inherently distributed and scalable: each UAV solves a low-dimensional local problem whose complexity does not increase with the size of the network, as neighbor information enters only through the cost function.

Following the framework of cooperative path-following developed in \cite{Antonio1, Reza1, kaminer2017time, Others}, we formulate the time-critical coordination problem as a "time synchronization" problem (see Section \ref{sec:prelim} for more details). This approach is inspired by the pioneering work of Leslie Lamport on time synchronization \cite{time-clock}. The key concept underpinning this framework is the decoupling of space and time in the general problem formulation, thereby reducing the multi-dimensional coordination problem to a one-dimensional consensus problem. This significantly reduces computational complexity and communication overhead, facilitating efficient real-time implementation. In  \cite{Reza1,  Coo-path-2013, 7065327, kaminer2017time, KH-23} the consensus is achieved using a PI controller. The exponential stability of time coordination was established for networks connected in an integral sense. Despite the advantages of PI-based methods \cite{kaminer2017time}, they heavily rely on apriori mission planning, which weakens its adaptability and applicability in dynamic environments, particularly in satisfying operational requirements and time-varying objectives of the mission as they change.  To overcome these limitations, in this paper, we reformulate the coordination problem within a game-theoretic framework, employing the virtual time concept (see \cite{kaminer2017time}) as a consensus parameter. In our setting, UAVs engage strategically to reach agreement on the coordination parameter. Compared to previous PI-based methods \cite{time-clock}, our approach provides greater generality and flexibility in accommodating operational constraints and dynamically evolving mission objectives. In particular, explicit constraints such as UAV speed and acceleration limits can be directly incorporated into the optimization problem, ensuring feasibility, unlike in PI-based methods, \cite{kaminer2017time}, where feasibility is more difficult to guarantee. Moreover, the optimization problem also enables the direct incorporation of complex, time-varying mission specifications—such as energy efficiency and collision avoidance (see Section \ref{remark-2})—into the optimization framework, further enhancing the adaptability and practicality of the proposed solution. Furthermore, our approach is distributed by its nature; hence, it maintains low per‐agent computational and communication overhead, ensuring scalable and robust coordination even in large‐scale, UAV networks (see Table~\ref{tab:data} in Section \ref{sec:sim}). Hence, this paper proposes a computationally efficient, scalable, and agile coordinated path-following method for UAV systems. A qualitative comparison between optimization-based and PI-based coordination methods is summarized in Table~\ref{tab:comparison}.

\textcolor{blue}{
\begin{table}[t]
\centering
\caption{Comparison between optimization-based (MPC) and PI-based coordination methods.}
\label{tab:comparison}
\scriptsize
\setlength{\tabcolsep}{3pt}
\renewcommand{\arraystretch}{1.1}
\begin{tabular}{lcc}
\toprule
\textbf{Feature} & \textbf{Optimization-based} & \textbf{PI-based} \\
\midrule
 Path-following error & Yes &  Yes\\
Explicit constraint handling & Straightforward & Depend on parameters choice \\
Operational requirements  & Yes & No \\
Formal stability guarantees & Guaranteed & Depend on parameters choice
 \\
Computational requirement &  IPOPT solver  & None
 \\
Communication load &  Vector $\mathcal{N}_i\times K$  & Vector $\mathcal{N}_i$ 
\\
Ease of implementation &  \checkmark & \checkmark \\
\bottomrule
\end{tabular}
\label{comparison}
\end{table}}
Figure~\ref{block:Diag} illustrates the architecture of the proposed solution. Unlike prior works \cite{Antonio1, Reza1, kaminer2017time} that employed a PI control law without addressing optimality, our approach introduces an Optimal Temporal Coordination block based on a game-theoretic formulation, where UAVs coordinate through constrained optimization, ensuring both feasibility and responsiveness to evolving mission demands.
 The main contributions of this paper are summarized as follows. We propose a novel game-based scalable coordination algorithm, where the Nash equilibrium of the strategic interactions among agents ensures system-wide coordination and successful task execution.
 We prove the existence of a Nash equilibrium and establish the exponential convergence of the solution under ideal conditions. 
  We develop an MPC-based algorithm to handle realistic constraints, including path-following errors, communication failures, and dynamic environmental conditions.
  
In contrast to classical Nash equilibrium existence results (see, for example, \cite{Bashar-1}, \cite{Pavel2012GameTheory}, and \cite{NE-Review}), we establish the existence of a continuous-time, infinite-horizon Nash equilibrium in an infinite-dimensional action space with constraints. A related result—establishing the existence of a Nash equilibrium with bounded control inputs and stability—was previously shown in \cite{NE-Non} for the discrete-time case using Lyapunov stability techniques.  In this paper, we make a significant progress by proving both the existence and exponential stability of a continuous-time, infinite-horizon Nash equilibrium with constraints, based on the analysis of the associated Euler–Lagrange equations—an approach that, to the best of our knowledge, is used for the first time in this setting. In particular, our coordination problem leads to double integrator dynamics, resulting in a fourth-order system of Euler–Lagrange differential equations, which we study by deriving its explicit solution. Furthermore, by leveraging this explicit form and applying appropriate transversality conditions, we establish the exponential stability of the equilibrium.

Regarding Nash equilibrium seeking algorithms, there are  existing results assuming  a fixed communication network topology (see, for example \cite{Shamma-g, Dusan_NE_1, Bashar-NE, Jeff-num-2017}). These methods often rely on standard consensus protocols or gradient-based dynamics in static environments, limiting their applicability in dynamic or time-varying settings. 
Recent works have begun to relax this assumption, considering time-varying network structures. For instance, the algorithms in \cite{Larca-Pavel-2017, NE-seeking-2021} explore convergence under dynamically evolving graphs, relying on conditions like joint or integral connectivity. However, these methods typically employ first-order dynamics and gradient-based updates.
 In  this paper, leveraging  MPC, 
 we propose an algorithm (see Section \ref{sec:sim}) that converges to  Nash equilibrium in the case of time-varying network, even in the presence of communication failures. Furthermore, the algorithm takes into account the path following errors, making it robust to environmental disturbances such as wind (see Section \ref{sec:sim1} for simulations). Nevertheless, each MPC step solves a simple optimization problem with the UAV’s virtual time as the only decision variable, avoiding optimization over high-dimensional UAV states and significantly reducing the computational cost. In the cooperative path-following framework, once the spatial trajectory is fixed during mission planning, the coordination layer operates solely on the virtual-time variable with simple upper and lower bounds, while neighbor information and path-following errors enter only as constants (see Problem~\ref{Discrete}). As a result, the method remains adaptive, robust, and scalable—its computational time is unaffected by the number of neighboring UAVs (see Table~\ref{tab:data})—making it suitable for real-time implementation. All computation times reported for each simulation scenario (see Section \ref{sec:sim}) are on the order of $10^{-2}$, further evidencing the effectiveness of the proposed method.

The remainder of this paper is structured as follows: Section \ref{sec:prelim} introduces the necessary notations and provides an overview of the preliminaries, including path planning, path following, and coordination. Furthermore, in Section \ref{sec:prelim}, we introduce an infinite horizon game problem with a discount factor. In Section \ref{sec:main},  by leveraging the system of Euler-Lagrange equations (a coupled fourth-order differential equations), we prove the existence of a Nash equilibrium for the corresponding unconstrained problem under ideal conditions. Moreover, we derive an explicit expression for the Nash equilibrium and prove that
the equilibrium
trajectory exhibits exponential convergence; see Proposition \ref{prop-exp-first}. Eventually, examining the coefficients of the explicit solution, we prove that there exists a discount rate such that the solution to the unconstrained problem is also solution to the constrained one; see   Theorem \ref{theorem-exp-stab}. Section \ref{sec:sim} outlines the general algorithm derived from the theoretical results, and Section \ref{sec:sim1} validates the algorithm's reliability and performance through simulations in challenging and realistic scenarios. The video footage showcasing the simulations and the flight experiments is available at the following \href{https://www.youtube.com/watch?v=yaI_HdUjJj4&t=77s}{link}.

\section{Preliminaries}\label{sec:prelim}

We consider a class of time-critical cooperative missions that can be formulated as  ``time synchronization'' problems (see, for example, \cite{time-clock, kaminer2017time}). 
We assume that for a given mission a path planner generates a desired and feasible trajectory (possibly optimal for the mission) for each UAV.  Next, we assume that each UAV is equipped with a robust low-level controller for tracking purposes. The coordination task is then reduced to a ``time synchronization'' problem. Although the proposed framework is general and can be applied to various heterogeneous UAV systems (multirotor, fixed-wing, VTOL, etc.), 
we focus on systems consisting of quadrotors for simplicity. 
Throughout the remainder of the paper, the term \emph{UAV} refers to a quadrotor unless stated otherwise.
 Section~\ref{Path-Planning-and-Following} discusses the constraints associated with offline trajectory-generation, while Section~\ref{Time-Coordination-for-Simultaneous-Arrival} discusses constraints on coordination parameters and defines the objectives underlying the time-coordination formulation.

\begin{figure}
    \centering
    \includegraphics[width=1\linewidth]{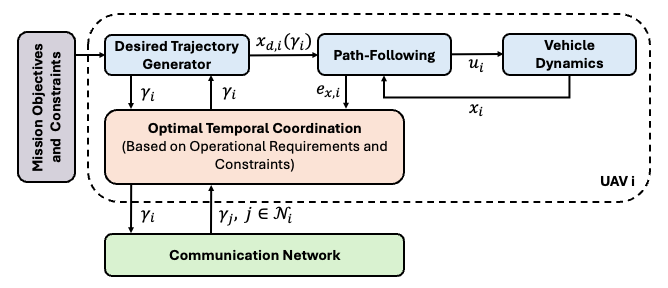}
    \caption{ Game-theoretic cooperative path-following approach block diagram.}
    \label{block:Diag}
\end{figure}

\subsection{Path Planning and Following}\label{Path-Planning-and-Following}

Consider a system of $N$  UAVs involved in the mission, $N \in \mathbb{N}$. For each UAV,
 the desired trajectory is given as a function $x_{d, i}:[0, t_{d, i}^*] \rightarrow \mathbb{R}^3$, $t_{d, i}^* > 0$. For all $t \in [0, t_{d, i}^*]$, the trajectories satisfy the constraints  defined by the $i^{\rm th}$ UAV's minimum and maximum linear velocity limits, and maximum linear acceleration limits \cite{kaminer2017time,cichella20133d}:
\begin{equation}\label{def-vel-constaints}
0\leq v^i_{\min}<v^i_{d, \min } \leq\left|\left|\dot{x}_{d, i}(t)\right|\right|\leq v^i_{d, \max }<v^i_{\max},
\end{equation}
where $v^i_{\min}$ and $v^i_{\max}$ are the $i^{\rm th}$ vehicle's possible minimum and maximum speeds, respectively, while, $v^i_{d,\min}$ and $v^i_{d,\max}$ are the vehicle's minimum and maximum speeds for the given mission. Similarly, 
\begin{equation}\label{def-acc-constaints}
\left|\left|\ddot{x}_{d, i}(t)\right|\right|\leq a^i_{d,\max} < a^i_{\max},
\end{equation}
with $a^i_{\max}$ being the $i^{\rm th}$ UAV's possible maximum acceleration and $a^i_{d,\max}$ being the maximum acceleration for the mission.

Time-critical mission specifications are encoded into the desired trajectories of the UAVs by the planner.  By properly defining these trajectories, it is possible to design missions with varying characteristics, such as simultaneous arrival and/or sequential auto-landing. The analysis of such systems is similar, and hence, here, we consider only the simultaneous arrival case. All UAVs are expected to complete their assigned paths under ideal conditions sharing a common nominal predefined final arrival time $t_d^* \in \mathbb{R}^{+}$:
$
t_{d, i}^*=t_d^*$, $i=1, \ldots, N.
$
Next, we assume that the trajectories  maintain spatial separation throughout the mission:
\begin{equation}\label{def-sp-sep}
\min _{\substack{i, j=1, \ldots, n \\ i \neq j}}||x_{d, i}(t)-x_{d, j}(t)||^2 \geq E^2>0, \, \text{for all } t\geq 0,
\end{equation}
where $E>0$ denotes the prescribed minimum allowable distance between any two UAV trajectories.

\textit{Remark 1.} In Section~\ref{remark-2},  we introduce an additional collision-avoidance term within the optimization framework. This term explicitly enforces safe separation between UAVs, thereby eliminating the need for the spatial separation constraints in~\eqref{def-sp-sep}. Consequently, the proposed formulation allows UAVs to follow intersecting or closely aligned trajectories while ensuring collision-free operation.

Various approaches exist for trajectory generation, each for specific objectives;   e.g., optimal control-based methods (minimizing energy consumption or path length), waypoint-based trajectory generation (connects waypoints using smooth polynomials), and minimum-snap trajectory generation algorithm (see \cite{Beard2012, 4310229,  Dubins1957OnCO, 5980409}).

We assume that each UAV is equipped with a path-following controller that ensures that the UAV follows its desired trajectory, \( x_{d,i}(t) \), or the updated trajectory based on coordination.  Examples include PID controllers \cite{salih2010flight}, geometric controllers \cite{5717652}, and adaptive path-follower \cite{7963104}, to name a few.

\subsection{Time Coordination for Simultaneous Arrival}\label{Time-Coordination-for-Simultaneous-Arrival}

To formulate the simultaneous arrival problem, we follow \cite{7065327, kaminer2017time}.
Let $\gamma_i:\mathbb{R}^+\rightarrow[0,t_d^*]$, $i=1,\dots,N$, map the clock time $t$ to the mission time $t_{d,i}=\gamma_i(t)$, where $t_{d,i}$ parameterizes a predefined desired trajectory.
The function $\gamma_i$ is referred to as the \emph{virtual time} of the $i$th UAV and represents an internal coordination variable that governs the progression of the UAV along its path. Unlike real time, $\gamma_i$ can be adjusted dynamically, enabling reparameterization of the desired trajectory as
\begin{equation}\label{def-reparamit}
x_{\gamma,i}(t):=x_{d,i}(\gamma_i(t)),
\end{equation}
where $x_{\gamma,i}$ denotes the virtual target to be tracked by the UAV.
This formulation decouples temporal evolution from spatial motion, allowing coordination through appropriate regulation of the virtual time. Specifically, When $\dot{\gamma}_{i}(t) = 1$, the UAV travels at the desired pace. When $\dot{\gamma}_{i}(t) > 1$, the UAV moves faster than the desired pace.
 When $\dot{\gamma}_{i}(t) < 1$, the UAV moves slower than the desired pace.
Thus, $\gamma_i$ serves the role of a consensus parameter in our problem. 
Based on the physical limitations and mission requirements of the UAVs, as defined in \eqref{def-vel-constaints} and \eqref{def-acc-constaints}, we derive general bounds for the consensus parameter of each UAV. Specifically, from \eqref{def-reparamit}, we have
 \begin{equation}\label{eq-deriv-x_id_beta}
\begin{split}
\dot{x}_{\gamma, i}(t)&=\dot{x}_{d, i}\left(\gamma_i(t)\right)\dot{\gamma}_i(t),\\
\ddot{x}_{\gamma, i}(t)&=\ddot{x}_{d, i}\left(\gamma_i(t)\right)\dot{\gamma}_i(t)^2+\dot{x}_{d, i}\left(\gamma_i(t)\right)\ddot{\gamma}_i(t).
\end{split}
\end{equation}
The first equation along with \eqref{def-vel-constaints} implies that the minimum and the maximum values of all admissible parameters $\gamma_i$ must satisfy the following constraint (for more details see \cite{7065327}):
\begin{equation}\label{def-constraits-beta-deriv}
\frac{v^i_{ \min }}{v^i_{d, \min }}	 \leq \dot{\gamma}^i_{\min } \leq \dot{\gamma}_i(t)\leq \dot{\gamma}^i_{ \max }\leq \frac{v^i_{ \max }}{v^i_{d, \max }}.
\end{equation}
These inequalities ensure that the UAVs maintain forward motion throughout the mission (due to $\dot{\gamma}_{i}(t) \geq 0$ condition).
On the other hand, by \eqref{def-acc-constaints} and the second equation of \eqref{eq-deriv-x_id_beta}  it follows that $\dot{\gamma}^i_{ \max }=\max_{t}\{\dot{\gamma}_{i}(t)\}$ and  $\ddot{\gamma}^i_{ \max }=\max_{t}\{|\ddot{\gamma}_{i}(t)|\}$ should satisfy  

\begin{equation}\label{def-constraits-beta-double-deriv}
\ddot{\gamma}^i_{ \max } v^i_{d, \max }+(\dot{\gamma}^i_{ \max })^2 a^i_{d, \max } \leq a^i_{\max }.
\end{equation}
Note that since $v^i_{d, \max }<v^i_{\max}$ and $a^i_{d,\max} < a^i_{\max}$, there exist   $\gamma_i$ parameters such that   inequalities in \eqref{def-constraits-beta-deriv} and \eqref{def-constraits-beta-double-deriv} are satisfied. For more details on above derivation see \cite{7065327}.

The \textbf{time coordination} of the system is achieved when the consensus  parameters are synchronized (see  \cite{kaminer2017time} and references within); that is,
\begin{equation}\label{def-cooperation1}
\gamma_i(t)-\gamma_j(t)=0, \quad \text{for all } i, j \in \{1, 2, \dots, N\}.
\end{equation}
This condition ensures that all UAVs reach their respective goal positions simultaneously.
Furthermore, to ensure that the UAVs maintain a predefined \textbf{desired speed} profile, the derivatives of the virtual times should satisfy:
\begin{equation}\label{def-cooperation2}
\dot{\gamma}_{i}(t)-1=0, \quad \text{for all } i \in \{1, 2, \dots, N\}.
\end{equation}
\vspace{-1cm}
\subsection{Problem Formulation}\label{sec:prob}
Next, we introduce a game-theoretic approach that facilitates time-critical cooperative mission execution. For the theoretical analysis, we make the following assumptions regarding the UAV system and its dynamics:

\begin{hyp}\label{hyp:1}
    Communication network of UAVs is fully connected; that is, $\mathcal{N}_i=\{1, \ldots, N\}$ for all $i=1, \ldots, N$, where $\mathcal{N}_i$ is set of UAVs with which $i^{th}$ UAV can exchange information.
\end{hyp}
\begin{hyp}\label{hyp:2}
    The motion of each UAV is governed by ideal path-following kinematics,
    meaning that the UAV perfectly tracks its assigned geometric path without deviation.
\end{hyp}
\begin{remark}\label{remark-no-path}
Assumption~\ref{hyp:2} implies that the desired trajectory reparameterizations through the corresponding virtual times $\gamma_i$, together with the linear constraints in \eqref{def-constraits-beta-deriv} and \eqref{def-constraits-beta-double-deriv}, still yield feasible trajectories. Hence, no path-following error occurs under this assumption.
\end{remark}

While these assumptions simplify the theoretical analysis (see Section \ref{sec:ex}), they do not limit the practical applicability of the proposed approach. In Section \ref{sec:sim}, we introduce an algorithm that incorporates path-following errors and communication failures, enabling application in real-world scenarios. As shown in \cite{N1, N2}, by using cascaded inner-loop outer-loop structure for flight control applications, the uncertainties in system dynamics can be handled by a variety of robust control methods.

We begin by introducing the admissible set of the consensus parameter (virtual time) for each UAV. Due to unexpected events, UAVs may begin their missions at different time instants, 
which corresponds to different initial virtual times, i.e., $\gamma_i(0) = \gamma_i^0$ for $i = 1, \dots, N$. Along with the inequalities in \eqref{def-constraits-beta-deriv}, \eqref{def-constraits-beta-double-deriv}, we define the admissible sets:
\begin{equation*}
\mathcal{A}^{0}_i := \Big\{ \gamma_i \in H^2_{w,\alpha}((0, \infty)): 
	\gamma_i(0) = \gamma^0_i\geq 0, \,\dot{\gamma}_i(0) = 1
\Big\},
\end{equation*}
and 
\begin{equation*}
\mathcal{A}^{2,\alpha}_i := \Big\{ \gamma_i \in \mathcal{A}^{0}_i: 
	\dot{\gamma}_i \geq 0, \,   
\|\dot{\gamma}_i\|_{L^\infty} \leq V^i_1, \,  \|\ddot{\gamma}_i\|_{L^\infty} \leq V^i_2
\Big\},
\end{equation*}
where  $H^1_{w,\alpha}((0, \infty))=\{g\in H^2_{loc} ((0, \infty)): \int_{0}^{\infty} e^{-\alpha t} (g^2+\dot{g}^2+\ddot{g}^2) \dt<\infty\}$ is weighted Sobolev space (for a detailed discussion of Sobolev and weighted Sobolev spaces see for example \cite{ForSobolev1}, \cite{ForSobolev2} and \cite{ForSobolev3}) with weight $e^{-\alpha t}$.

\begin{problem}\label{prob-first} Consider a system of $N$ UAVs. Each UAV (agent)  over $\gamma_i\in\Aa^{2,\alpha}_i$ seeks to minimize
	\begin{equation}\label{def-cost}
		\begin{split}
J_i(\gamma_i,\gamma_{-i})=\int_{0}^{\infty}&e^{-\alpha t}\Big(w_1\left(\dot{\gamma}_i-1\right)^2\\&+\frac{w_2}{|\mathcal{N}_i|}\sum_{j \in \mathcal{N}_{i}}(\gamma_i-\gamma_{j})^2 +w_3\ddot{\gamma}_i^2\Big)\dt,
		\end{split}
	\end{equation}
where $\gamma_{-i}=(\gamma_1,\dots,\gamma_{i-1},\gamma_{i+1},\dots,\gamma_N)$, $\mathcal{N}_{i}$ is set of neighbors of $i^{th}$ UAV and $w_1, w_2, w_3 \in \mathbb{R}_0^+$ are nonnegative constants satisfying 
$w_1 + w_2 + w_3 = 1$, which determine the relative importance of the corresponding penalty terms.
\end{problem}

The cost function in \eqref{def-cost} corresponds to simultaneous arrival and  maintaining a desired mission pace with the following terms.
     \textbf{Pace penalty:} $(\dot{\gamma}_i - 1)^2$ penalizes deviations from the desired pace, and the parameter $w_1$ determines its relative importance.
     \textbf{Coordination penalty:} $(\gamma_i - \gamma_j)^2$ penalizes discrepancies with neighboring agents, and the parameter $w_2$ determines its relative importance.
     \textbf{Control effort penalty:} $(\ddot{\gamma}_i)^2$ penalizes excessive control inputs, and the parameter $w_3$ determines its relative importance.
 
We refer to the solution of Problem \ref{prob-first} as Nash equilibrium. 
\begin{definition} The vector function $\gamma^*=(\gamma^*_1,\dots,\gamma^*_N)$ is  Nash equilibrium of Problem \ref{prob-first}, if for all $ \gamma_{i}\in \Aa^{2,\alpha}_i$
\begin{equation}
J_i(\gamma^*_i,\gamma^*_{-i})\leq J_i(\gamma_i,\gamma^*_{-i}).
\end{equation}
\end{definition}

\section{Main Approach}\label{sec:main}
In this section,  we prove the existence and exponential stability of the solution to Problem \ref{prob-first}. Furthermore, we introduce an MPC based algorithm that approximates the solution of Problem \ref{prob-first}.
\vspace{-0.55cm}
\subsection{Existence of Solution and Exponential Stability}\label{sec:ex}
 Under ideal system conditions, we isolate and analyze the core features of the game-theoretic approach. More precisely, to establish the existence and exponential stability of the solution to Problem \ref{prob-first}, we first prove the existence and exponential stability of the solution to the corresponding unconstrained problem. Then, we determine the parameter $\alpha$ (discount rate) such that the solution to the unconstrained problem is also a solution to the constrained problem. We then conclude that, for the chosen $\alpha$, Problem \ref{prob-first} has a solution that exhibits exponential convergence with respect to the equilibrium trajectory.

 The unconstrained problem can be formulated as follows: 

\begin{problem}\label{prob-exp-stab} Consider a system of $N$ UAVs. Each UAV (agent) seeks to minimize
	\begin{equation}\label{def-cost-0}
I^\alpha_{\bm{\gamma}}[\gamma_i]=\int_{0}^{\infty}e^{-\alpha t}\Big(w_1\dot{\gamma}_i^2+\tfrac{w_2}{N}\sum_{j=1}^{N}(\gamma_i-\gamma_{j})^2 +w_3\ddot{\gamma}_i^2\Big)\dt
	\end{equation}
    over $\gamma_i\in\mathcal{B}^{0}_i:=\{\gamma_i\in H^2_{w,\alpha}((0, \infty)):(\gamma_i
    +t)\in\mathcal{A}^{0}_i \}$.
\end{problem}

Note that in $\gamma_i(t)$ the argument $t$ refers to the \emph{physical time} of the UAV. This is distinct from the variable $t$ used inside the optimization problem~\eqref{def-cost-0}, which serves as a \emph{dummy variable of integration} in the optimization.
The following result proves the existence and exponential stability of the solution to Problem \ref{prob-exp-stab} and provides an explicit expression. 
\begin{proposition}
    \label{prop-exp-first}
    Let $\alpha,w_3>0$, $w_1, w_2\geq 0$, and $w_1+w_2+w_3=1$ and suppose that Assumptions \ref{hyp:1}-\ref{hyp:2} hold. Then, there exists a unique $\bm{\gamma}^*=(\gamma_1^*,\dots,\gamma_N^*)\in \prod_{j=1}^{N}\mathcal{B}^{0,\alpha}_j$ solving Problem \ref{prob-exp-stab}. Moreover, letting $W := w_1^2 - 4w_2 w_3$ and 
$z_i^{r}(t) := H_{1i}^r + H_{3i}^r e^{\mu_3 t}$ 
(for $r \in \{0, +, -\}$), 
the solution takes the following explicit form:
\begin{equation}\label{eq-sol-General-3}
\begin{split}
\gamma_i^*(t) = z_i^{0}(t)+C^0_{1i} e^{\mu_1^0 t}  + C^0_{2i}t e^{\mu^0 t},
\end{split}
\end{equation}
when $W=0$; 
\begin{equation}\label{eq-sol-General-2}
\begin{split}
\gamma_i^*(t) = z_i^{+}(t)+C^+_{1i} e^{\mu^+_1 t}  + C^+_{2i} e^{\mu^+_2 t},
\end{split}
\end{equation}
when $W>0$; and
\begin{align}
    \label{eq-explicit-solution}
        \gamma_i^*(t)=z_i^{-}(t)+e^{\mu^-_1 t} (C^-_{i1} \cos(\nu_1 t) + C^-_{i2} \sin(\nu_1 t)),
    \end{align}
    when $W<0$. The constants $\mu_3$, $\mu_1^r$, $\mu_2^+$, $H_{1i}^r$, $H_{2i}^r$, 
$C_{i1}^r$, $C_{i2}^r$, and~$\nu_1$ (for $r \in \{0, +, -\}$) 
depend only on $\alpha$, $w_1$, $w_2$, $w_3$, and the initial conditions, and satisfy $\mu_3, \mu_1^r, \mu_2^+ < 0$.
   \end{proposition}
    \begin{proof}
    Here, we present the general idea of the proof; for more details see Appendix \ref{appen:A}.
Suppose that  $\bm{\gamma}^*=(\gamma_1^*,\dots,\gamma_N^*)\in \prod_{j=1}^{N}\mathcal{B}^{0}_j$ is a solution to  Problem \ref{prob-exp-stab}. Then, by the  definition of Nash equilibrium and the convexity of the integrand of \eqref{def-cost-0} it follows that $\gamma_i^*$ is the unique minimizer of the following optimization problem
       \begin{equation}\label{eq-min-prob}
\begin{split}
    I^\alpha_{\bm{\gamma}^*}[\gamma_i^*] &= \min_{\gamma_i \in \mathcal{B}^{0}_i} I^\alpha_{\bm{\gamma}^*}[\gamma_i].
\end{split}
\end{equation}
        To examine the behavior of $\gamma_i^*$,   $i=1,\dots,N$,  at infinity, we consider Euler-Lagrange equations of \eqref{eq-min-prob} 
\begin{equation}\label{eq-5-poly-1}
\gamma_i^{(4)*}-2\alpha  \dddot{\gamma}_i^*+(\alpha^2-\tfrac{w_1}{w_3}) \ddot{\gamma}_i^* +  \tfrac{\alpha w_1}{w_3} \dot{\gamma}_i^* +\tfrac{w_2}{w_3N}\sum_{j=1}^{N}  (\gamma_i^* - \gamma_j^*)=0,
\end{equation}
with the transversality conditions 
\begin{equation}\label{eq-transver-p}
\begin{split}
   & \lim_{T\to\infty} e^{-\alpha T}\ddot{\gamma}_i^*(T)=0,\\
  &  \lim_{T\to\infty} e^{-\alpha T}(\dot{\gamma_i}^*(T)+\alpha \ddot{\gamma}_i^*(T) - \dddot{\gamma}_i^*(T))=0.
\end{split}
\end{equation}
Using~\eqref{eq-transver-p} together with the boundary conditions in~$\mathcal{B}_i^{0}$, 
we solve the system of Euler--Lagrange equations in~\eqref{eq-5-poly-1} and find the unique solution. 
Depending on the sign of the parameter $W$, 
we obtain~\eqref{eq-sol-General-3}, \eqref{eq-sol-General-2}, 
or~\eqref{eq-explicit-solution}, respectively. 

To prove the existence of the solution to Problem \ref{prob-exp-stab}, we note that since $\gamma_i^*$ solves the Euler-Lagrange equation in \eqref{eq-5-poly-1}, and the variational problem in \eqref{eq-min-prob} is quadratic, we have that  $\gamma_i^*$ is the minimizer of \eqref{eq-min-prob}. Subsequently, $\bm{\gamma}^*$ is a Nash equilibrium of Problem \ref{prob-exp-stab}.
        \end{proof}
  Next, we prove the existence and exponential convergence of the solution to the constrained problem, i.e. Problem \ref{prob-first}. 
\begin{theorem}\label{theorem-exp-stab}
   Let the weights satisfy one of the following conditions: $w_1^2-4w_2w_3<0$ or $
                w_1^2-4w_2w_3\geq 0$ with $w_1=w_3 O(\alpha)$. Suppose that Assumptions \ref{hyp:1}-\ref{hyp:2}  hold.  Then, for any initial condition 
    $\gamma^0 = (\gamma_1^0, \dots, \gamma_N^0)$ 
    and given the physical constraints of the UAVs 
    (denoted by $V_1^i$ and $V_2^i$), 
    there exists $\alpha > 0$ such that 
    Problem~\ref{prob-first} admits a unique solution, 
    and the corresponding trajectory exhibits 
    \emph{exponential convergence} to the equilibrium trajectory.
\end{theorem}
\begin{proof} In the detailed proof of Proposition~\ref{prop-exp-first} (see Appendix~\ref{appen:A}), 
we explicitly track the dependence of the constants 
$\mu_3$, $\mu_1^r$, $\mu_2^+$, $H_{1i}^r$, $H_{2i}^r$, 
$C_{i1}^r$, $C_{i2}^r$, and~$\nu_1$ (for $r \in \{0, +, -\}$) 
on the parameters $\alpha$, $w_1$, $w_2$, and~$w_3$. Using these dependencies as $\alpha\to\infty$, we have
\begin{align}
      &H^{r}_{1i}= \gamma_i^0+O(\tfrac{1}{\alpha}),\mu_3=O(\tfrac{1}{\alpha}),\mu^r_1=O(\tfrac{1}{\alpha}),\mu^+_2=O(\tfrac{1}{\alpha})\notag\\\label{eq-limits}
     &\nu_1=O(\tfrac{1}{\alpha}),H^{r}_{i3}=O(\tfrac{1}{\alpha}),C^{r}_{1i}=O(\tfrac{1}{\alpha}), C^{r}_{2i}=O(\tfrac{1}{\alpha}),
\end{align}
where $r \in \{0, +, -\}$. On the other hand, the explicit  solution to Problem \ref{prob-exp-stab} (see \eqref{eq-sol-General-3}, \eqref{eq-sol-General-2}, \eqref{eq-explicit-solution}) is a combination of uniformly bounded  functions in $t$: $\cos(\nu_1 t),\quad \sin(\nu_1 t),\quad e^{\mu_1 t},\quad e^{\mu_3 t}$.
This  with \eqref{eq-limits} implies that taking large enough $\alpha$ for the explicit solution to Problem \ref{prob-exp-stab}, we obtain $(\gamma_i^*+t)\in \mathcal{A}^{2,\alpha}_i$. Therefore, $(\bm{\gamma}^*+t)$ is a solution to the constrained problem, Problem \ref{prob-first}, as well. By Proposition \ref{prob-exp-stab}, it follows that the solution exhibits exponential convergence with respect to the equilibrium trajectory.
\end{proof}
\vspace{-0.5cm}
\subsection{Algorithm}
\label{sec:sim}
In this part, we develop an algorithm based on Problem \ref{prob-first} to achieve time coordination among UAVs. The algorithm employs the idea of MPC, as a natural discrete-time continuation of the original optimization problem \ref{prob-first} (see Fig.~\ref{fig:problem_flow}). Beyond this conceptual consistency, the MPC formulation offers several practical advantages under non-ideal conditions, such as variations in network topology and path-following errors. In particular, MPC provides a systematic mechanism for enforcing input and state constraints at each time step while optimizing over a finite prediction horizon, which is an essential capability for UAV coordination, where actuation limits, such as constraints on velocity and acceleration, must be satisfied. Moreover, MPC formulation enables the accommodation of time-varying objectives of individual UAVs and autonomous adaptation to environmental changes and external disturbances. For instance, the effectiveness of this feature is illustrated through the demonstrated collision avoidance behaviour. To clarify the relationship between the continuous-time coordination problem (Problem~\ref{prob-first}) and the discrete-time MPC formulation presented in this section (Problem~\ref{Discrete}), we present a flow chart in Fig.~\ref{fig:problem_flow}.
\begin{figure}[t]
\centering
\begin{tikzpicture}[
  block/.style={rectangle, draw, rounded corners, minimum width=2.8cm, minimum height=0.5cm, text centered, font=\footnotesize},
  arrow/.style={-{Stealth}, thick, font=\scriptsize},
  node distance=0.5cm
]

\node[block, align=left] (p1) {Problem 1:\\ Continuous-time coordination};
\node[block, below=of p1, align=left] (p2) {Discretization:\\ Infinite-horizon formulation};
\node[block, below=of p2, align=left] (p3) {Problem 3:\\ MPC optimization};

\draw[arrow] (p1) -- (p2) node[midway,right]{Discretize time, reformulate cost/constraints};
\draw[arrow] (p2) -- (p3) node[midway,right]{Receding horizon, real-time solution};
\end{tikzpicture}
\caption{Flow chart showing the transition from the continuous-time coordination problem (Problem~\ref{prob-first}) to the discrete-time MPC problem (Problem~\ref{Discrete}).}
\label{fig:problem_flow}
\end{figure}
In the formulation of Problem \ref{prob-first}, the communication network is static and completely connected; i.e., each UAV is connected to all others involved in the mission. The proposed algorithm relaxes these assumptions. 
We assume that the communication network remains static during each MPC step but may change between steps. The communication network change is incorporated into the cost function.  Furthermore, the algorithm considers the path-following errors in the constraints that may arise when each UAV follows its desired trajectory operating in an uncertain and dynamic environment.

\begin{problem}\label{Discrete}  Let $h>0$ be the time step. To approximate the virtual time at time $t_{k}=kh$,  each UAV solves the following minimization problem:
\begin{equation*}
\begin{cases}
\min_{y^{k}_{i}} J_i\left(y^{k}_{i},\bar{s}^{k}_{-i}\right)\\
s^{k}_{i\tau+1}=s^{k}_{i\tau}+h\ell^{k}_{i\tau}+\dfrac{h^2}{2} u^{k}_{i\tau},\quad \tau=0,\dots,K-1,\\
\ell^{k}_{i\tau+1}=\ell^{k}_{i\tau}+hu^{k}_{i\tau},\quad \tau=0,\dots, K-1,\\
s^{k}_{i0}= s^{k-1}_{i1}-\alpha_{i}^{k},\quad \ell^{k}_{i0}= \ell^{k-1}_{i1},\\
\dot{\gamma}^i_{ \min }\leq\ell^{k}_{i\tau}\leq\dot{\gamma}^i_{ \max },\quad \tau=0,\dots, K,\\
\left|u^{k}_{i\tau}\right|\leq\ddot{\gamma}^i_{ \max },\quad\tau=0,\dots, K-1,
\end{cases}
\end{equation*}
for $k=1,2,\dots$. The optimization variables are $y^{k}_{i}=[s^{k}_{i},\ell^{k}_{i},u^{k}_{i}]$,
where $s^{k}_{i}=[s^{k}_{i0},s^{k}_{i1},\dots,s^{k}_{iK}]$ and
\begin{equation*}
\begin{split}
\ell^{k}_{i}=[\ell^{k}_{i0},\ell^{k}_{i1},\dots,\ell^{k}_{iK}],\quad
u^{k}_{i}=[u^{k}_{i0},u^{k}_{i1},\dots,u^{k}_{iK-1}].
\end{split}
\end{equation*}
\end{problem} 
The cost function $J_i$ is defined as follows:
\begin{equation}\label{cost1}
J_i(y^{k}_{i},\bar{s}^{k}_{-i})=\sum_{\tau=1}^{K}( \ell^{k}_{i\tau}-1)^2+F_i(s^{k}_{i},\bar{s}^{k}_{-i})+\sum_{\tau=0}^{K-1}{u^{k}_{i\tau}}^2,
\end{equation}
with $\bar{s}^{k}_{-i} = \left[\bar{s}^{k}_{1},\dots,\bar{s}^{k}_{i-1},\bar{s}^{k}_{i+1},\dots,\bar{s}^{k}_{N}\right]$
and $F_i(s^{k}_{i},\bar{s}^{k}_{-i})=\sum_{j\in \mathcal{N}_{ik}}\sum_{\tau=1}^{K}(s^{k}_{i\tau}-\bar{s}^{k}_{j\tau})^2$, 
where $\mathcal{N}_{ik}$ denotes the neighborhood of $i^{\rm th}$ UAV, which is the set of UAVs that communicate with $i^{\rm th}$ one at a time $t_{k}$.
All terms of the cost function in \eqref{cost1} have been discussed in the continuous case; see Section \ref{sec:prob}.  The newly added correction term $\alpha_{i}^{k}$  appears in initial conditions and takes into account the path-following error:
\begin{equation}\label{def-a}
\alpha_{i}^{k} =\alpha_{i}^{k}(x_{i}(t_{k}))= \beta\tfrac{\left(x_{\gamma,i}(t_{k})-x_{i}(t_{k})\right)^T\dot{x}_{\gamma,i}(t_{k})}{\left|\left|\dot{x}_{\gamma,i}(t_{k})\right|\right|+\delta},
\end{equation}
where $\beta$ and $\delta$ are positive parameters and $x_{i}(t_{k})$ is the $i^{\rm th}$ UAV's actual position. The correction term $\alpha_{i}^{k}$ is negative if the UAV's actual location projection on the desired trajectory is ahead of the desired position $x_{\gamma,i}(t_{k})$ and positive otherwise. Including $\alpha_{i}^{k}$ in initial conditions causes the UAV to slow down to reduce the forward overshoot and to accelerate to catch up with the desired trajectory.

\begin{remark}\label{remark-path-yes}
Unlike the ideal setting of Problem~\ref{prob-first} (see Assumption~\ref{hyp:2} and Remark~\ref{remark-no-path}), 
Problem~\ref{Discrete} explicitly accounts for path-following errors. 
Such errors may arise when the trajectory reparameterized by the virtual time—subject only to the linear constraints in 
\eqref{def-constraits-beta-deriv} and \eqref{def-constraits-beta-double-deriv}—is not dynamically feasible. 
These effects are incorporated through the correction term~$\alpha_i$. 
Provided that each UAV employs a sufficiently accurate low-level path-following controller, the proposed coordination method remains effective. 
This is validated by the simulation results in Section~\ref{subsection-non-ideal-path}.
\end{remark}

The parameter $\delta$ in the correction term $\alpha_{i}^{k}$ is introduced to prevent division by zero; it can be set as a positive constant close to $1$ without affecting the overall performance. The parameter $\beta$ regulates the influence of the path-following error on the coordination process. Increasing $\beta$ accelerates time coordination but may cause constraint violations if chosen excessively large.

At next  time instance $t_{k+1}$ each UAV shares the computed $[s^{k}_{i1},\dots,s^{k}_{iK}]$ (solution to Problem \ref{Discrete}) with the UAVs that are from its neighborhood
  $\mathcal{N}_{ik+1}$ at time instance $t_{k+1}$. For the $j^{\rm th}$ UAV from $\mathcal{N}_{ik+1}$, we have $\bar{s}^{k+1}_{j\tau}=s^{k}_{j\tau+1}$, $\tau=1,\dots,K-1$.  
The algorithm starts at time $t_{1}$ ($k=1$)  with  $s_{i0}^{1}=\gamma_{i}^{0}$, $l_{i0}^{1}=\dot{\gamma}_{i}^{0}$ and $\bar{s}_{j\tau}^{1}=\gamma_{j}^{0}+\tau h$, $\tau=1,\dots, K$, which are known due to initial information sharing between neighboring UAVs. Eventually, 
as approximations to $\gamma_{i}(t)$ and to its first-order and second-order derivatives, we set
\begin{equation*}
\begin{split}
\gamma_{i}(t_{k})\approx s^{k}_{i1},\quad \dot{\gamma}_{i}(t_{k})\approx \ell^{k}_{i1},\quad \ddot{\gamma}_{k}(t_{k})\approx u^{k}_{i1}.
\end{split}
\end{equation*}

The second term of the cost function \eqref{cost1} shows that the change in the communication can significantly impact the virtual time, particularly its second-order derivative. To mitigate and smoothen these effects, mainly in scenarios where communication depends on the distance between UAVs, we introduce the  smoothing function:
$\phi\left(z,p_{1},p_{2}\right)=1$, when $z<  p_{1}$,  $\phi\left(z,p_{1},p_{2}\right)=\eta(z)$, when $ p_{1}\leq z \leq  p_{2}$ and $ \phi\left(z,p_{1},p_{2}\right) =0$, when $p_{2}<z$. Here,
 $p_{1}<p_{2}$ are positive parameters and $0\leq\eta(z)\leq1$, such that $\phi\in C^{2}(\Rr^+)$ for any $p_{1}<p_{2}$.
The function $\phi$ quantifies the quality of communication between UAVs based on their relative distances. Communication is excellent if the distance between two UAVs is less than $p_{1}$. In contrast, communication is lost if the distance exceeds $p_{2}$.  For distances in the intermediate range $\left[p_{1},p_{2}\right]$, the quality of communication smoothly transitions between these two extremes.  The modified cost function is
\begin{multline}\notag
J_i(y^{k}_{i},\bar{s}^{k}_{-i})=\sum_{\tau=1}^{K}( \ell^{k}_{i\tau}-1)^2+F_i^{\phi}(s^{k}_{i},\bar{s}^{k}_{-i})+\sum_{\tau=0}^{K-1}(u^{k}_{i\tau})^2
\end{multline}
with
$F_i^\phi(s^{k}_{i},\bar{s}^{k}_{-i})
=
\sum_{j\neq i}\phi(d_{i,j}^{k,0},a,b)\sum_{\tau=1}^{K}(s^{k}_{i\tau}-\bar{s}^{k}_{j\tau})^2$,
where $d_{i,j}^{k,\tau}= ||x_{d,i}(s^{k}_{i\tau})-x_{d,j}(\bar{s}^{k}_{j\tau})||$.
Although the smoothing function is not convex, the modified cost function is quadratic because the function $\phi$ is treated as a constant in the cost function for each MPC step. Since it also does not impact the constraints,  Problem \ref{Discrete} remains a quadratic optimization problem. For the numerical analysis, we use equal weights, $w_1=w_2=w_3$, for all penalization terms. A sensitivity analysis of the weights, the path-following gain $\beta$ and MPC horizon $K$, is provided at \href{https://github.com/mikayel2/swarm_timecoord_mpc}{this
 link}.
\subsubsection{Collision Avoidance}\label{remark-2} The formulation of Problem \ref{Discrete} is agile and can be modified to achieve additional goals; for example, the cost function can be augmented with additional terms that  ensure collision avoidance:    
\begin{multline}\label{costphipsi}
J_i(y^{k}_{i},\bar{s}^{k}_{-i})=\sum_{\tau=1}^{K}( \ell^{k}_{i\tau}-1)^2+\sum_{\tau=0}^{K-1}{u^{k}_{i\tau}}^2\\+G_i(s^{k}_{i},\bar{s}^{k}_{-i})+F_i^{\phi,\psi}(s^{k}_{i},\bar{s}^{k}_{-i}).
\end{multline}
The third term of \eqref{costphipsi} is defined as follows:
\begin{equation}\label{collisionavoidance}
\begin{split}
G_{i}(s^{k}_{i},\bar{s}^{k}_{-i}) = &\sum_{j\neq i} \sum_{\tau=1}^{K}\tfrac{C_{i}\phi(d_{i,j}^{k,\tau},a,b)}{({d_{i,j}^{k,\tau}})^2}.
\end{split}
\end{equation}
This term becomes active when the distance between the UAVs falls below $b$, prompting the UAVs to adjust the speed and ensure collision avoidance. The effect of the collision avoidance term \eqref{collisionavoidance} is the highest when the distance is less than or equal to $a$. $C_{i}$ is a positive parameter that weighs the collision avoidance term impact for each UAV. Additionally, to ensure  time coordination, the term $F_i^{c m}$ is defined as follows:
\begin{equation*}
F_i^{\phi,\psi}(s^{k}_{i},\bar{s}^{k}_{-i}) = \sum_{j\neq i} \phi(d_{i,j}^{k,0},c,d)\psi(d_{i,j}^{k,0},a,b)\sum_{\tau=1}^{K}(s^{k}_{i\tau}-\bar{s}^{k}_{j\tau})^2,
\end{equation*}
where $a <b < c<d$ and $\psi(z,q_{1},q_{2})=0$, when $z< q_{1}$, $\psi(z,q_{1},q_{2})= \xi(z)$, when $ q_{1} \leq z \leq q_{2}$ and $
\psi(z,q_{1},q_{2})=1,q_{2}\leq z$. Here,
$q_{1}<q_{2}$ are positive parameters and $0\leq\xi(z)\leq1$,  such that $\psi\in C^{2}(\Rr^+)$. As a result, time coordination is achieved when the distance between the UAVs is greater than $a$. Furthermore, coordination is disregarded when the distance is less than $a$, and only collision avoidance is active.
Although the cost function \eqref{costphipsi} is not convex, the optimization problem is still simple enough to solve efficiently in real-time.  
\vspace{-0.55cm}
\section{Simulations}\label{sec:sim1}
 In this section, we demonstrate the effectiveness and applicability of the proposed method in complex, realistic scenarios.
\vspace{-1cm}
\subsection{Simulation Framework}
We use a high-fidelity simulation framework with multirotor UAV dynamics to demonstrate the performance of the algorithm. The simulation setup incorporates factors such as aerodynamic effects, actuator limitations, and environmental disturbances. The simulations validate the methodology and showcase its agility when deployed in various scenarios.

\subsubsection*{Software Setup} The simulations were performed using RotorPy \cite{folk2023rotorpy}, an open-source Python-based simulator designed for multirotor UAVs. In this work, we specifically utilize the Crazyflie multirotor model. RotorPy provides dynamic modeling of multirotor systems, including aerodynamic effects and nonlinear equations of motion.

The simulation scenarios were designed to explore different aspects of the proposed coordination algorithm. The MPC parameters used in all scenarios are as follows: the number of UAVs is fixed at $N=6$ , the prediction horizon is $K=10$, and the time step is $h=0.05 \mathrm{~s}$. The state constraints are $\gamma_{\min }^i=0$ and $\dot{\gamma}^i \in[0,2]$, while the control constraint is $\ddot{\gamma}^i \in[-6,6]$. In all simulations, we set $C_i=i+1$ in the collision-avoidance term (\ref{collisionavoidance}). Varying parameters for the scenarios are the communication terms $c$, $d$, and the collision avoidance terms $a$ and $b$, whose values are accordingly presented.

The algorithm, Alg. 1, employed in these simulations, is
designed as follows: the desired state for each UAV is computed, and then passed to the path-following controller that guides
the UAV. Subsequently,  the path-following error is evaluated, and the initial conditions of the Problem \ref{Discrete} are updated.
Then, based on the solution of Problem \ref{Discrete} the UAV transmits $[s^{k}_{i2},\dots,s^{k}_{iK}]$ to the UAVs from the neighborhood $\mathcal{N}_{ik+1}$. Afterwards, the desired state for the UAV is updated, and the process is repeated iteratively. This process results in the coordination of the multi-agent system.
\vspace{-0.5cm}
\begin{algorithm}
\caption{Multi-agent time-critical MPC}
\label{alg:algorithm}
\SetKwFunction{FGetCurrentPosition}{GetActualPosition}
\SetKwFunction{FSolveProblem}{SolveProblem3}
\SetKwProg{Fn}{Function}{:}{}
\SetKwFor{For}{for}{do}{end for}
\SetKwFor{While}{while}{do}{end while}
\SetKw{KwInitialize}{Initialize}

\KwInitialize number of agents $N$, desired trajectories $x_{di}$, state and control constraints $\dot{\gamma}^{i}_{\min}$, $\dot{\gamma}^{i}_{\max}$, $\ddot{\gamma}^{i}_{\max}$, initial conditions $\gamma_{i}^{0}$, $\dot{\gamma}_{i}^{0}$, $i=1,\dots,N$, prediction horizon $K$, total time $T$, time step $h$\;
Set $k=1$\;

\While{$kh \leq T$}{
    \For{$i \gets 1$ \KwTo $N$}{
        $x_{i}(kh) \gets \FGetCurrentPosition(x_{di}(s_{i0}^k))$\;
        Compute $\alpha_{i}^{k}(x_{i}(kh))$\;
        
        Update initial conditions\;
        $s_{i0}^{k}\gets s^{k-1}_{i1}-\alpha_{i}^{k}$, $l_{i0}^{k}\gets l_{i1}^{k-1}$\;
        
        $s_{i}^{k}$, $l_{i}^{k} \gets \FSolveProblem(s_{i0}^{k}, l_{i0}^{k}, \bar{s}_{-i}^{k})$\;
        
        Transmit $[s^{k}_{i2},\dots,s^{k}_{iK}]$ to UAVs from $\mathcal{N}_{i}^{k+1}$\;
        
        $\bar{s}^{k+1}_{i}\gets [s^{k}_{i2},\dots,s^{k}_{iK}]$\;
    }
    $k \gets k + 1$\;
}
\end{algorithm}
\vspace{-0.5cm}
\subsubsection{Ideal Communication and Ideal Path-Following}\label{sec:ideal_communication}
We consider an ideal scenario with perfect communication and no external disturbances, where all UAVs can exchange information freely. Under these conditions, six UAVs follow non-overlapping, heterogeneous trajectories in three-dimensional space. The UAVs aim to minimize their cost functions, with each starting its mission at a distinct time $\gamma_1^0 = 2$, $\gamma_2^0 = 1$, $\gamma_3^0 = 0$, $\gamma_4^0 = 3.5$, $\gamma_5^0 = 4$, and $\gamma_6^0 = 3$.  Note that only the 3rd UAV starts its mission precisely on schedule. The total mission lasts 36 seconds. The results demonstrate rapid synchronization among the UAVs, making $\gamma_1 = \gamma_2 = \gamma_3 = \gamma_4 = \gamma_5 = \gamma_6$ in around 4.15 seconds, as shown in Fig. \ref{fig:scenario1_gamma}. Moreover, the UAVs successfully adjust their velocities to converge to their desired speed profiles (see Fig. \ref{fig:scen1_gammadot}). The control input converges to zero by the 6.5$^{\rm th}$ second (Fig. \ref{fig:scen1_gammaddot}). Furthermore, the algorithm's efficiency is confirmed by a maximum MPC step execution time of 0.0157 seconds, ensuring it does not impose a significant computational burden.
The simulation results indicate that solutions of Problem \ref{prob-first} exhibit exponential convergence with respect to the equilibrium trajectory, which validates Theorem \ref{theorem-exp-stab}.
\begin{figure*}[ht]
    \centering
    \begin{subfigure}[b]{0.32\textwidth}
        \centering
        \includegraphics[width=\textwidth]{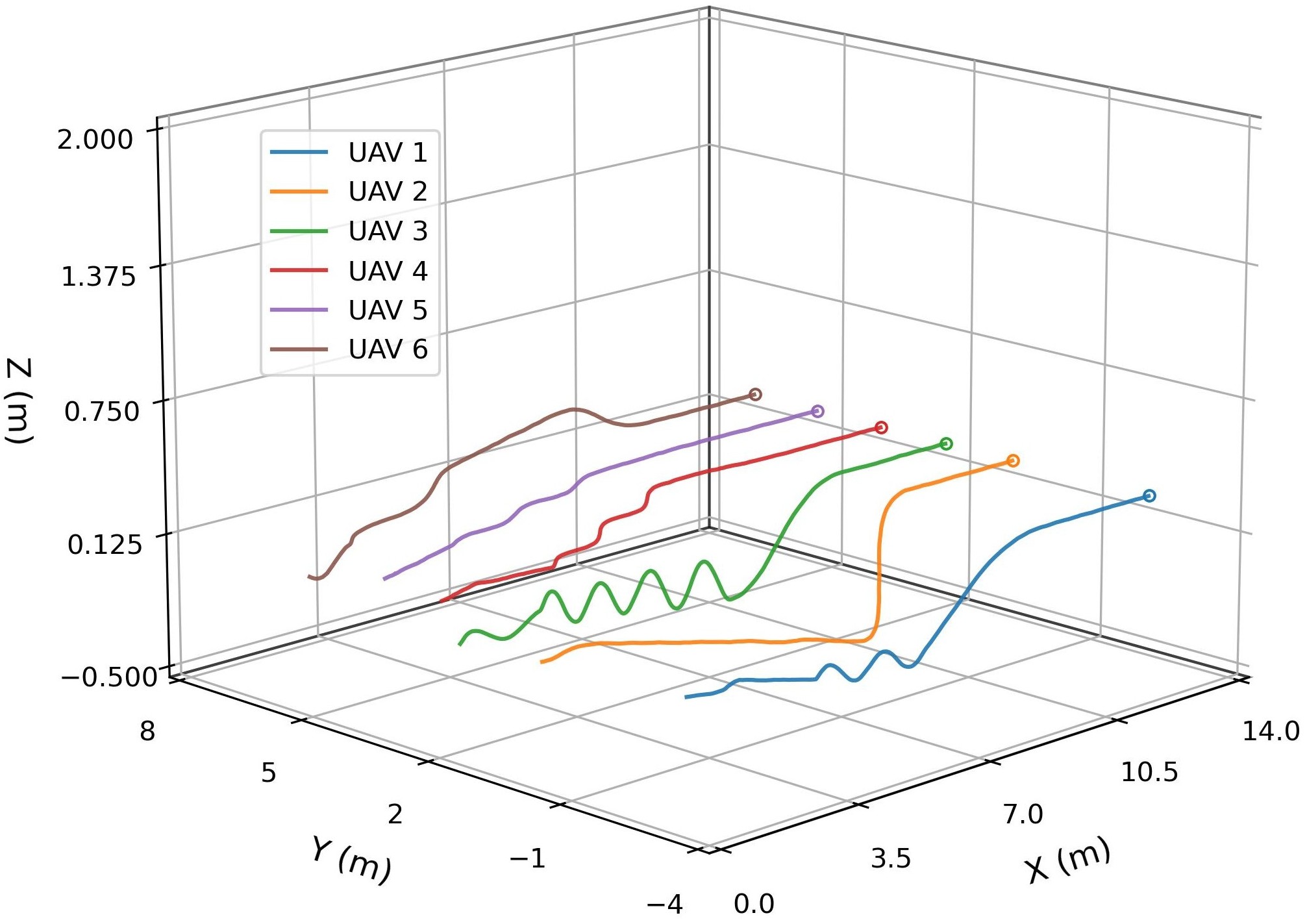}
        \caption{}
        \label{fig:scenario1_traj}
    \end{subfigure}
    \hfill
    \begin{subfigure}[b]{0.32\textwidth}
        \centering
        \includegraphics[width=\textwidth]{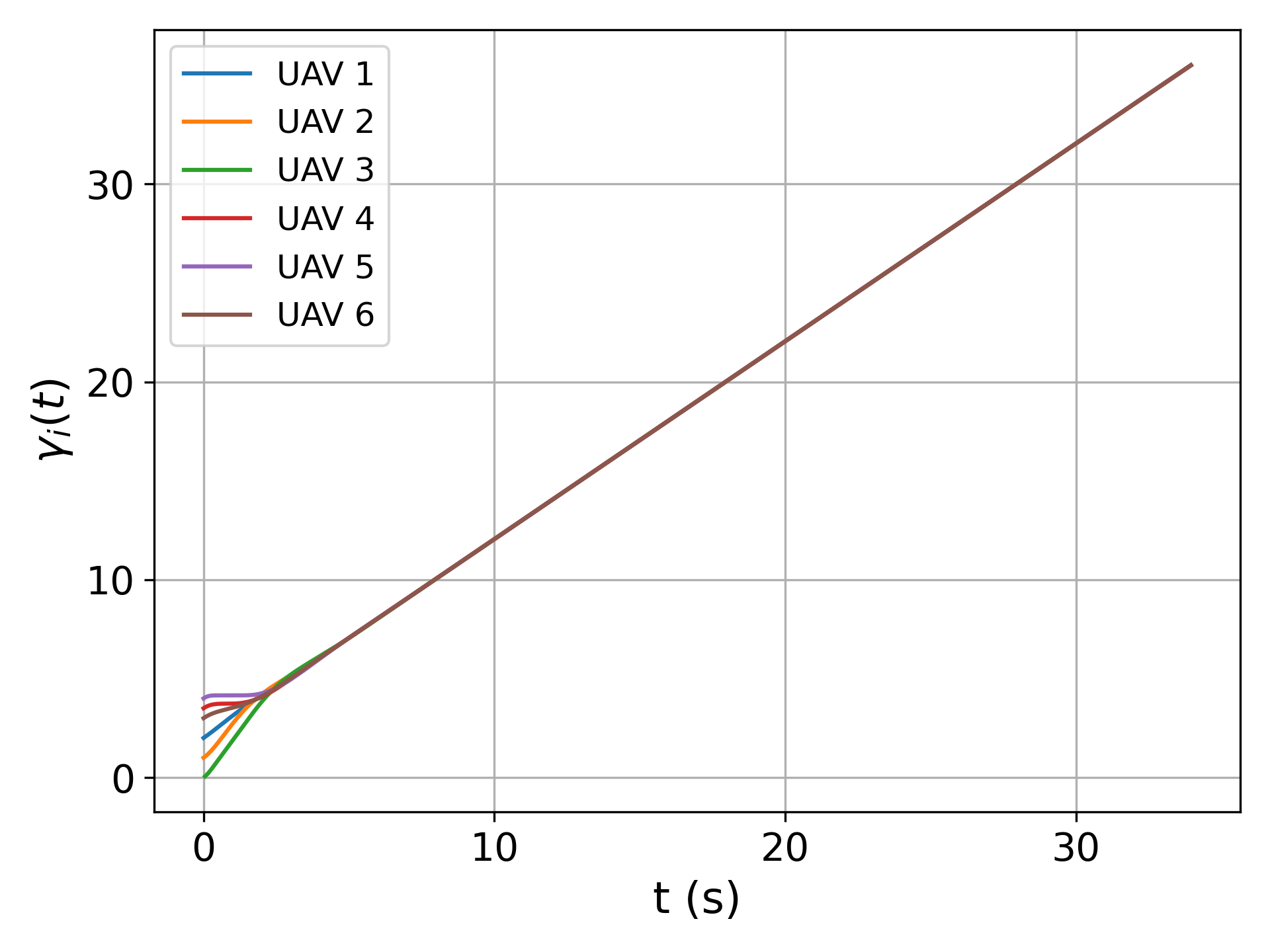}
        \caption{}
        \label{fig:scenario1_gamma}
    \end{subfigure}
    \hfill
    \begin{subfigure}[b]{0.32\textwidth}
        \centering
        \includegraphics[width=\textwidth]{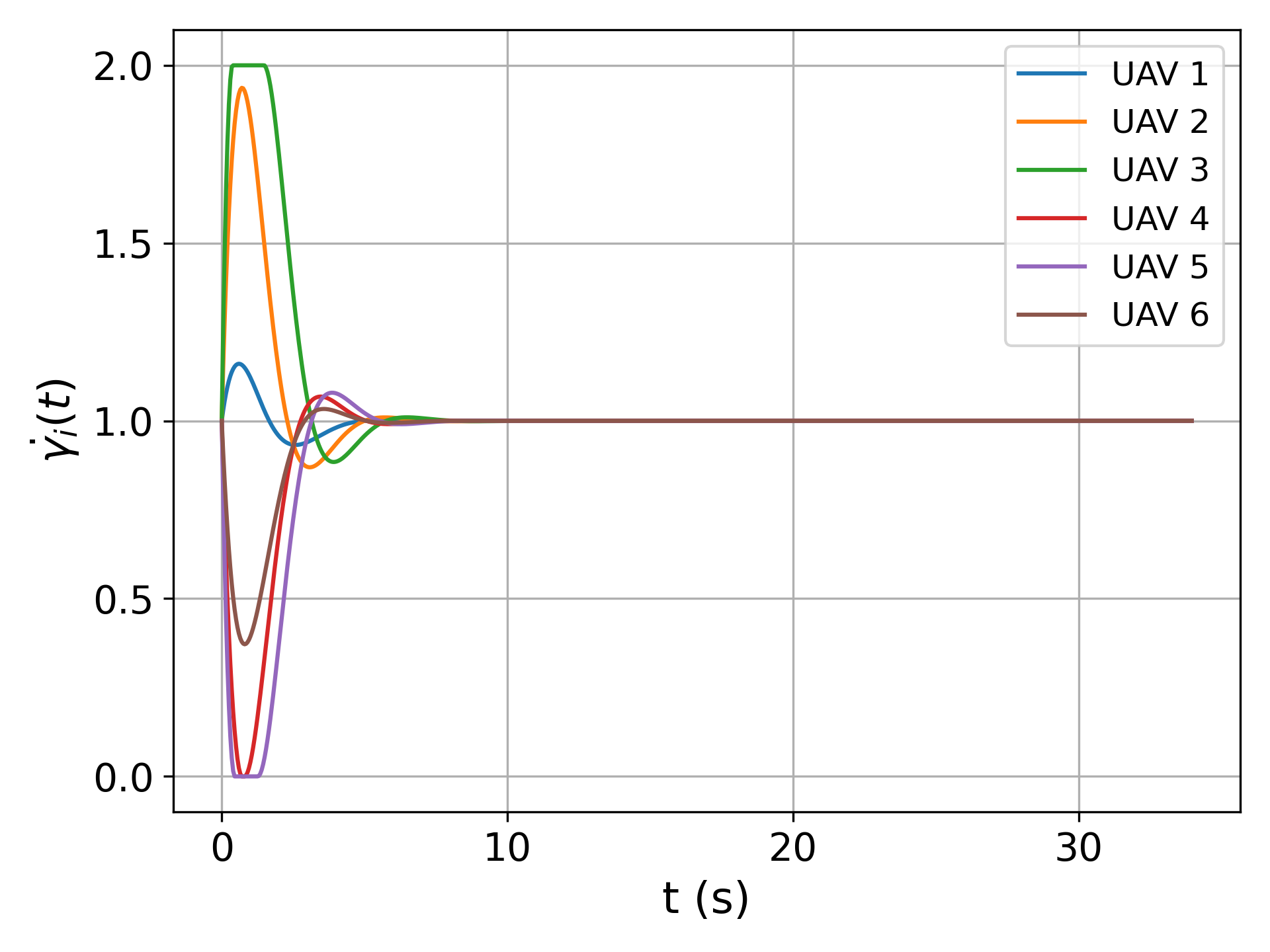}
        \caption{}
        \label{fig:scen1_gammadot}
    \end{subfigure}
    
    \caption{Ideal communication, ideal path following: (a) A top view of the actual trajectories followed by the UAVs under ideal conditions, with circles marking the final positions; (b) $\gamma_i$ over time; (c) $\dot{\gamma_i}$ over time.}
    \label{fig:three_scenario_plots}
\end{figure*}

\subsubsection{Non-Ideal Communication and Ideal Path-Following}{\label{scen2}}
Building on the conditions of the ideal communication scenario, this case introduces non-ideal communication among UAVs, where information exchange is limited. Two UAVs can share their predicted consensus parameters only if the distance between them is below a predefined threshold. The parameters for the communication term are set as follows: $c = 2.25$m, $d =4.5$m. At the beggining of the mission, UAV 1 and UAV 2 do not communicate with either UAV 5 or UAV 6. Despite intermittent communication interruptions and subsequent reconnections, the UAVs successfully achieve coordination, with the third and fourth UAVs serving as communication links between all UAVs. 

Compared to the ideal communication scenario, system synchronization is achieved in a longer time frame. Specifically, as shown in Fig.~\ref{fig:scenario1_gamma}, coordination takes approximately 4.15 seconds in an ideal communication setting. However, under the added challenging communication conditions, the coordination time in this case is almost 11 seconds (see Fig.~\ref{fig:scen2_gamma}). Along the same lines, the control input converges to zero at around 6.5 seconds in the ideal communication scenario (see Fig.~\ref{fig:scen1_gammaddot}), while under non-ideal communication conditions, it takes significantly longer—almost 12.5 seconds—to reach zero, as illustrated in Fig.~\ref{fig:scen2_gammaddot}. This delay illustrates the impact of poor communication on coordination, as it obstructs the agents' ability to quickly synchronize and stabilize the control input. Moreover, the presence of the smoothing function $\phi$  in the algorithm does not allow for sharp oscillations of the control input $\ddot{\gamma}$, which, in turn, prevents drastic changes in drone behavior. In this scenario, the maximum time taken for the MPC step calculation is 0.0166.

\begin{figure*}[ht]
    \centering
    \begin{subfigure}[b]{0.32\textwidth}
        \centering
        \includegraphics[width=\textwidth]{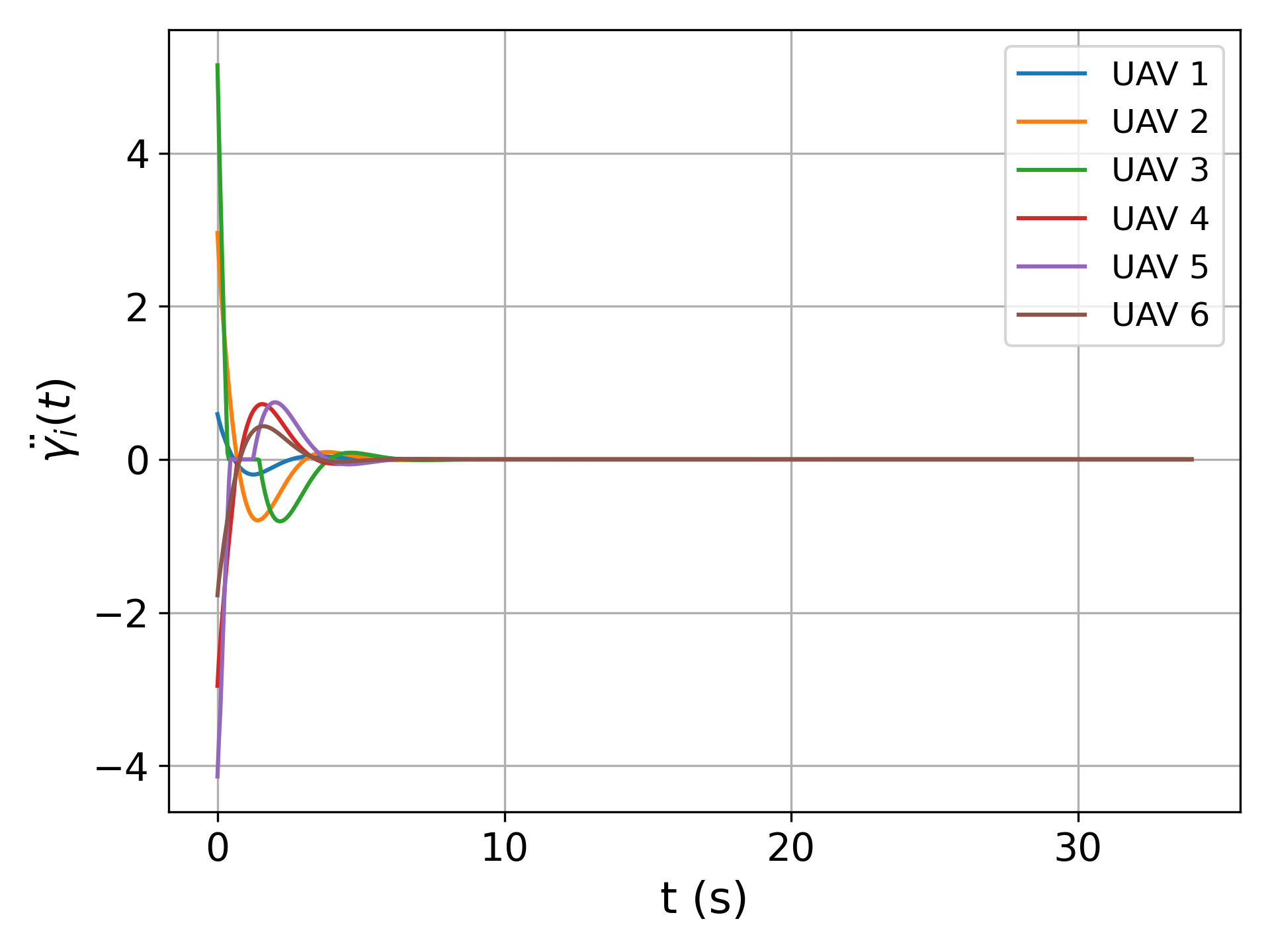}
        \caption{}
        \label{fig:scen1_gammaddot}
    \end{subfigure}
    \hfill
    \begin{subfigure}[b]{0.32\textwidth}
        \centering
        \includegraphics[width=\textwidth]{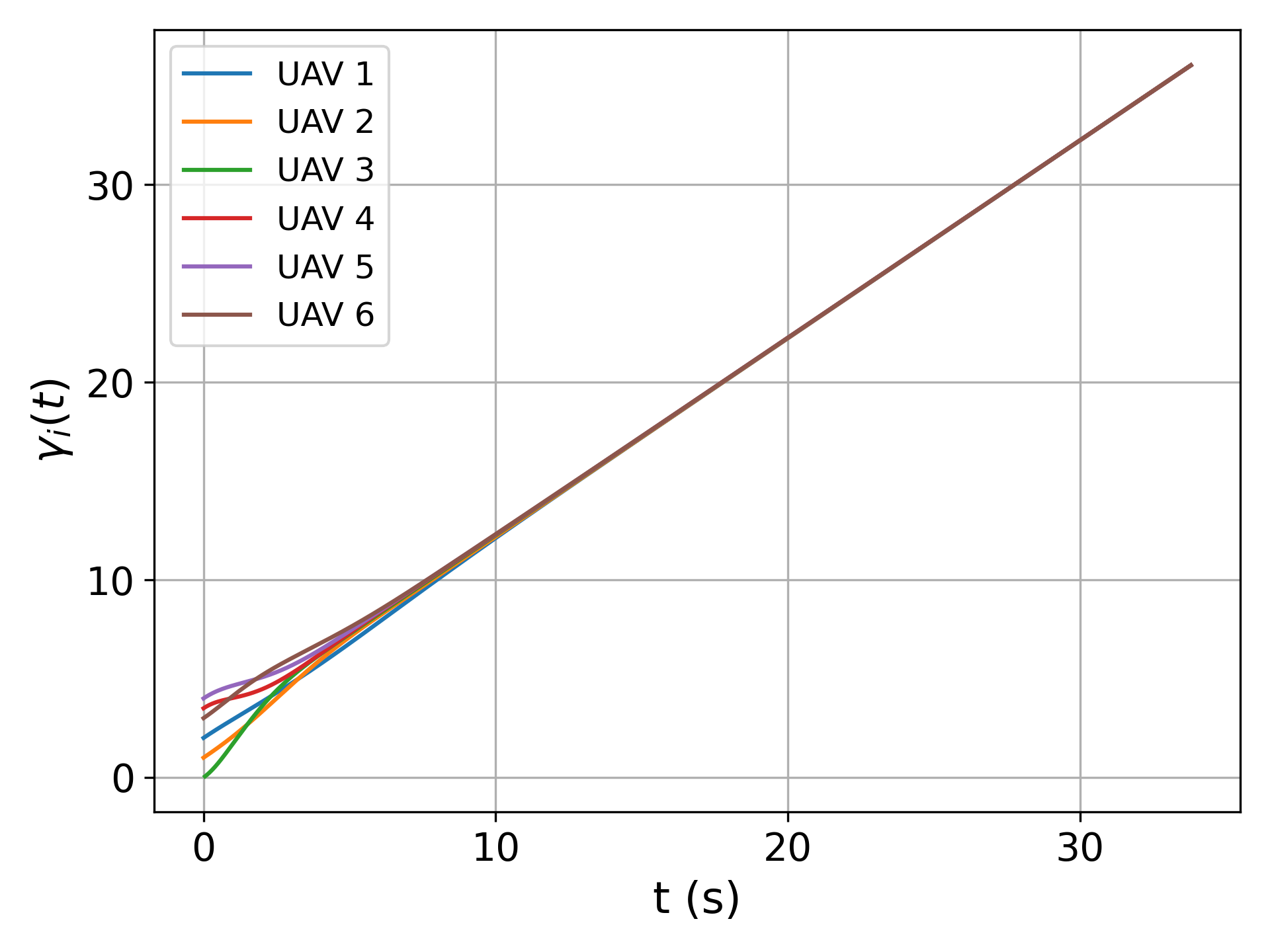}
        \caption{}
        \label{fig:scen2_gamma}
    \end{subfigure}
    \hfill
    \begin{subfigure}[b]{0.32\textwidth}
        \centering
        \includegraphics[width=\textwidth]{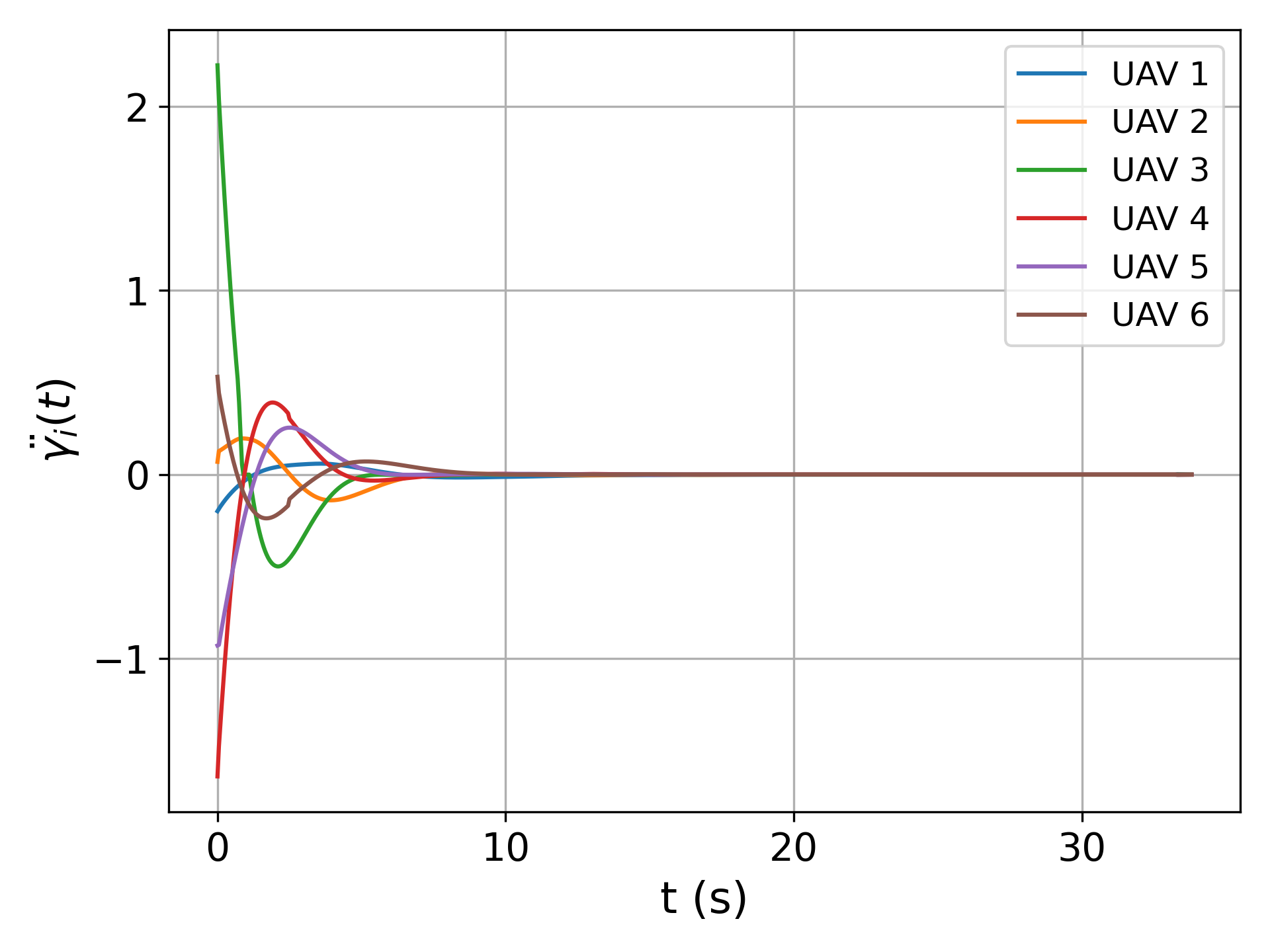}
        \caption{}
        \label{fig:scen2_gammaddot}
    \end{subfigure}
    
    \caption{ (a) Ideal communication, ideal path following, $\ddot{\gamma_i}$ over time, (b) Non-ideal communication, ideal path following, $\gamma_i$ over time, (c) Non-ideal communication, ideal path following, $\ddot{\gamma_i}$ over time.}
    \label{fig:scen1_scen2}
\end{figure*}

\subsubsection{Non-Ideal Communication and Non-Ideal Path-Following}\label{subsection-non-ideal-path}
In a non-ideal path following scenario, an additional wind disturbance is introduced along with partial communication failures. The wind disturbance affects the motions of all UAVs from the start of the 36-second simulation, lasting for a total of 18 seconds (half of the mission time). The disturbance begins with an initial wind speed of 7 m/s that decreases linearly, reaching 0 m/s precisely at the 18-second mark. The wind blows horizontally towards the negative y-axis direction (Fig. \ref{fig:scen3_traj}). This causes the UAVs to significantly shift from their desired trajectories. Additionally, the communication between UAVs is updated every 0.5 seconds, with each pair having a 70\% chance of an active communication link.

The wind-affected trajectories, depicted in Fig. \ref{fig:scen3_traj}, show the impact of the path-following term $\alpha_i^k$ (see \eqref{def-a}) introduced in the algorithm. We set the parameter $\delta = 1$ in the path-following error term. This term enables the UAVs to realign with their desired trajectories, ensuring that coordination happens during the mission. Synchronization is achieved at the 17.25-second mark, when the wind disturbance dissipates and the UAVs can ideally realign and return to their desired paths. However, communication failures delay the consensus time and affect the control input behavior. In the presence of wind disturbances and partial communication failures, the stabilization of the control inputs is notably delayed, eventually converging to zero at the 25-second mark, as shown in Fig. \ref{fig:scen3_gammaddot}. Despite this convergence, residual oscillations are still observed as a clear effect of the communication failures. In such cases, the UAVs exert more effort to stay on their designated paths. The maximum MPC step calculation time is measured to be 0.0192 seconds throughout the simulation.

\begin{figure*}[ht]
    \centering
    \begin{subfigure}[b]{0.28\textwidth}
        \centering
        \includegraphics[width=\textwidth]{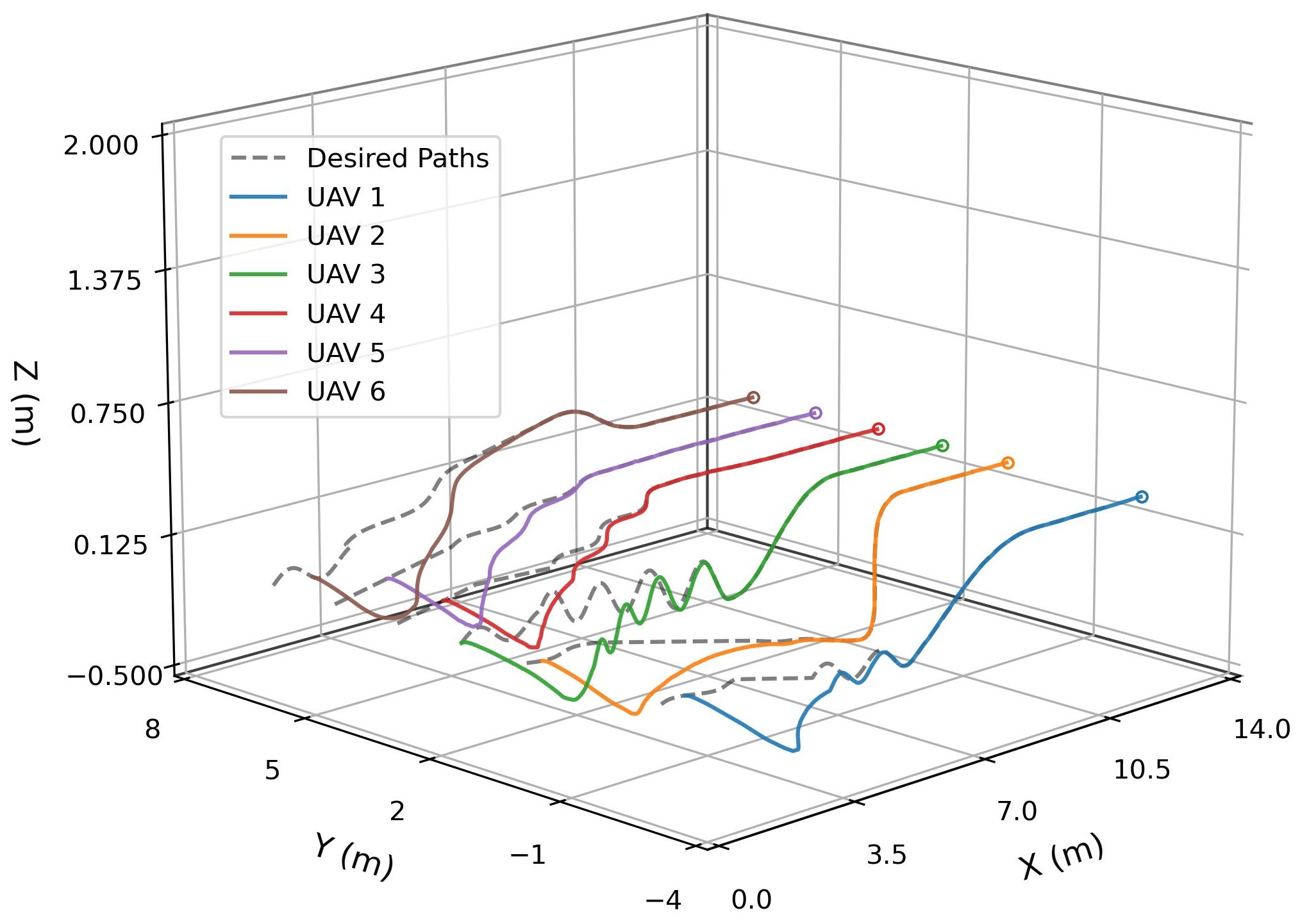}
        \caption{}
        \label{fig:scen3_traj}
    \end{subfigure}
    \hfill
    \begin{subfigure}[b]{0.28\textwidth}
        \centering
        \includegraphics[width=\textwidth]{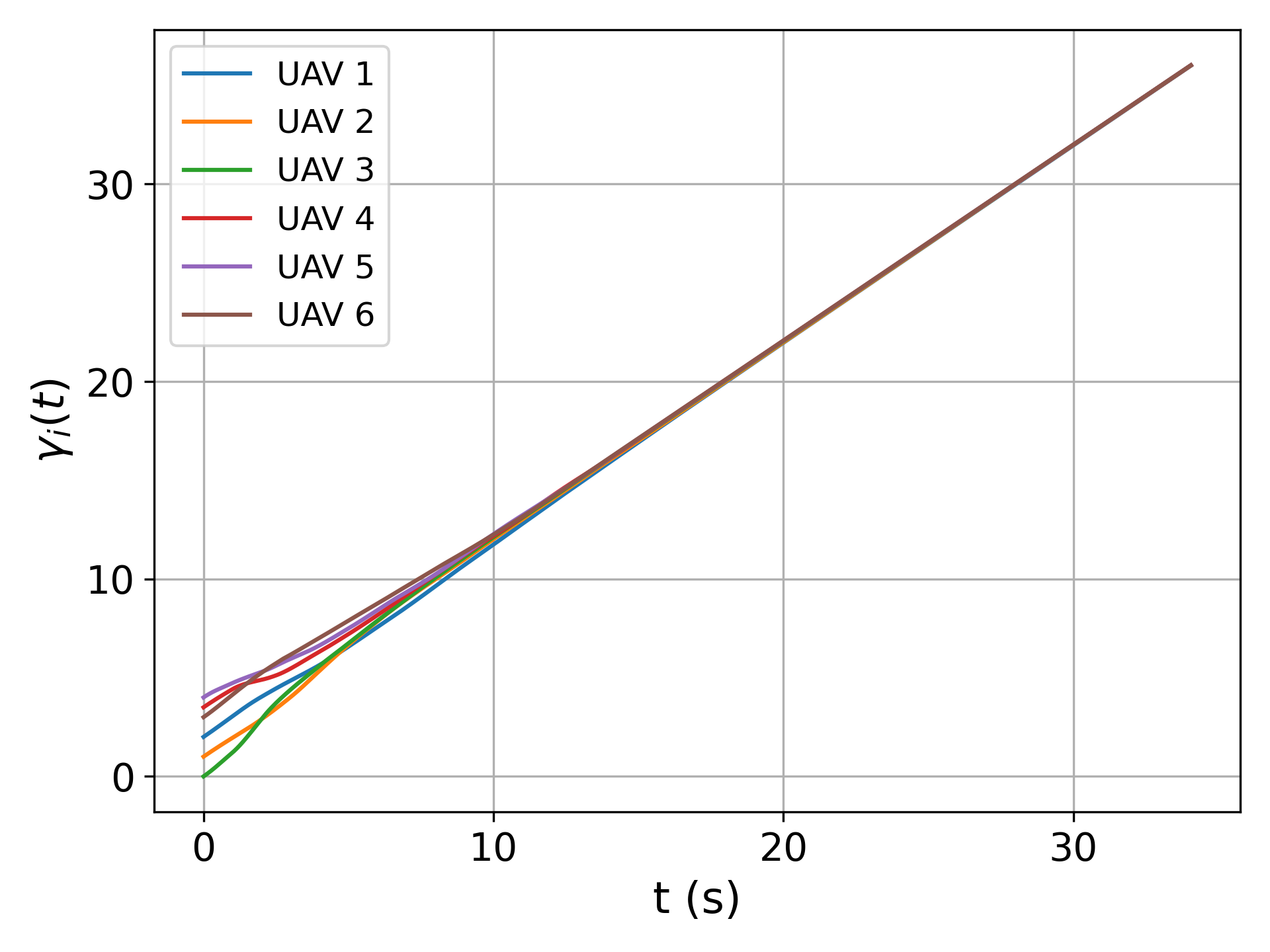}
        \caption{}
        \label{fig:scen3_gamma}
    \end{subfigure}
    \hfill
    \begin{subfigure}[b]{0.28\textwidth}
        \centering
    \includegraphics[width=\textwidth]{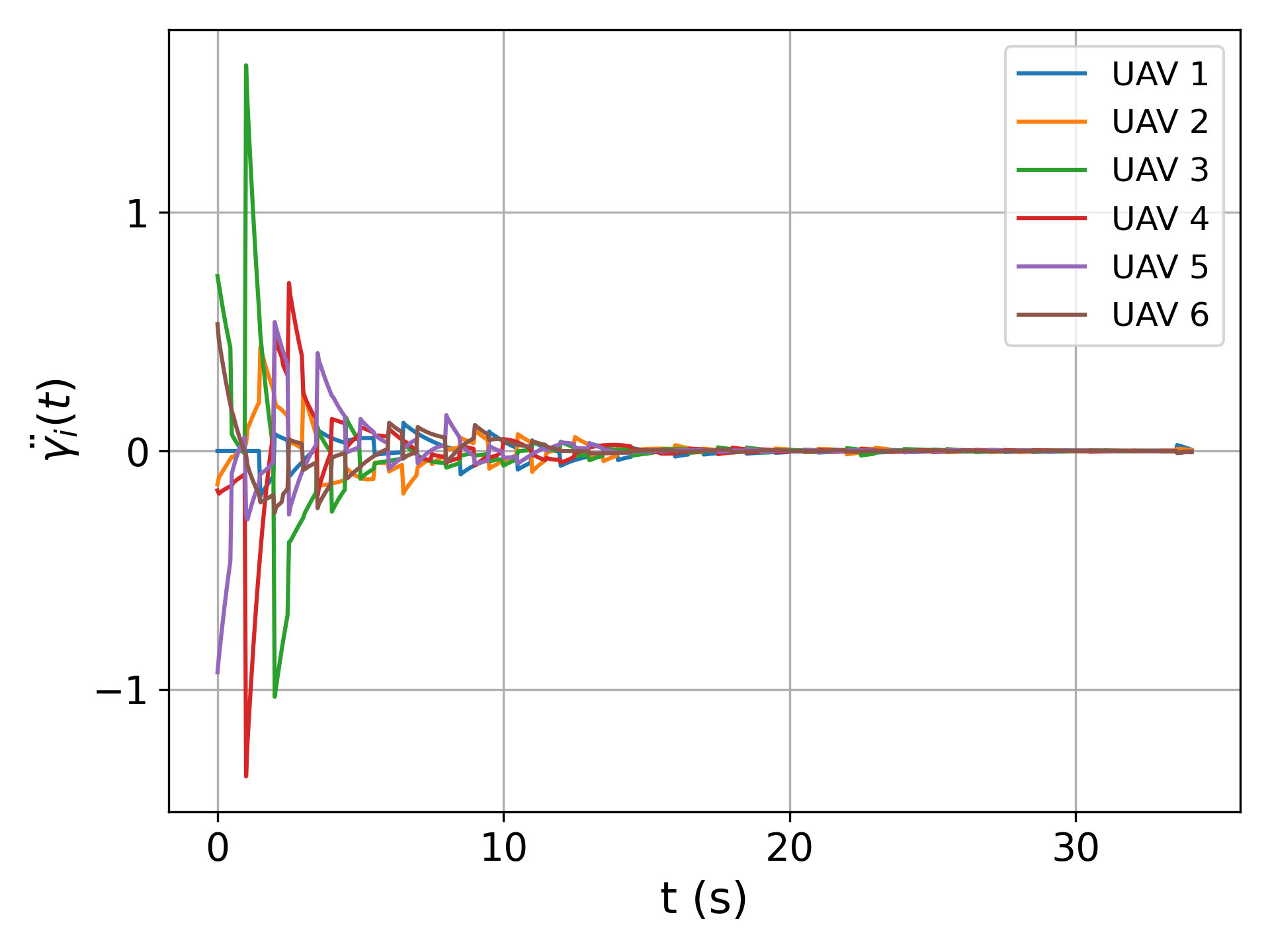}
        \caption{}
        \label{fig:scen3_gammaddot}
    \end{subfigure}
    
    \caption{Non-ideal communication, non-ideal path following: (a) A top view of the trajectories influenced by wind disturbance and partial communication failures, with circles marking the final positions; (b) $\gamma_i$ over time; (c) $\ddot{\gamma_i}$ over time.}
    \label{fig:scenario3}
\end{figure*}

\subsubsection{Collision Avoidance}\label{subsec-collison-avoid}
Based on the discussion in Section \ref{remark-2}, we demonstrate coordination along with collision avoidance, a scenario involving six UAVs following Lissajous trajectories intersecting at a single central point was implemented (Fig. \ref{fig:scen4_traj}). The trajectories are given by
${\displaystyle x(k)=X\sin(vk+\epsilon )}$ and ${\displaystyle y(k)=Y\sin(wk)}$, 
with the following parameter values: $X = 0.8$, $Y = 8$, $v = 0.5$, $w =0.25$, $\epsilon = 0$.
In this setup, no time delays are implemented; thus $\gamma_i^0 = \gamma_j^0 = 0$. The communication terms were set as follows: $c = 10$m, $d = 20$m, ensuring perfect communication throughout the mission. Collision avoidance parameters were set separately for each of the six UAVs: $a = [2.35, 2.5, 2.7, 2.8, 2.9, 3.0]$m and $b = [4.7, 5.0, 5.4, 5.6, 5.8, 6.0]$m.
The mission lasted 42 seconds, during which the UAVs completed three passes through their trajectory, successfully avoiding collisions at the intersection point and coordinating afterward. As illustrated in Fig. \ref{fig:scen4_gamma}, the UAVs avoided collisions between seconds 8 and 13, followed by complete synchronization between seconds 16 and 18. For further clarity, the minimal distance between any two UAVs was plotted over time, Fig. \ref{fig:min_distance}, which confirms that a minimum separation of at least 0.5 meters was consistently maintained throughout the mission, making sure that no collisions occurred during the simulation. During the mission, the maximum execution time for a single MPC step was observed to be 0.028 seconds. This increase in processing time over previous cases is due to the added computational demands of collision avoidance.

\begin{figure*}[ht]
    \centering
    \begin{subfigure}[b]{0.28\textwidth}
        \centering
        \includegraphics[width=\textwidth]{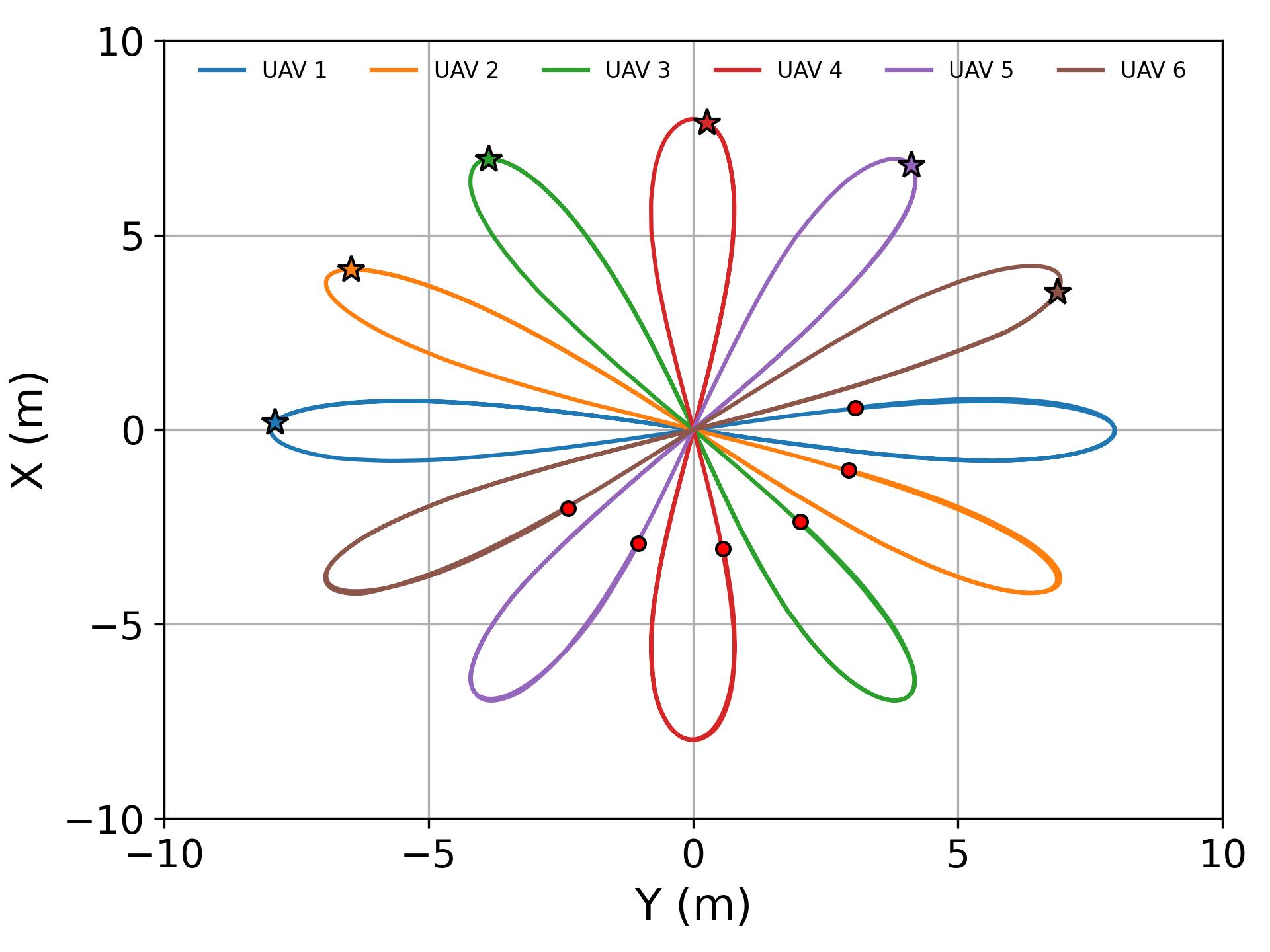}
        \caption{}
        \label{fig:scen4_traj}
    \end{subfigure}
    \hfill
    \begin{subfigure}[b]{0.28\textwidth}
        \centering
        \includegraphics[width=\textwidth]{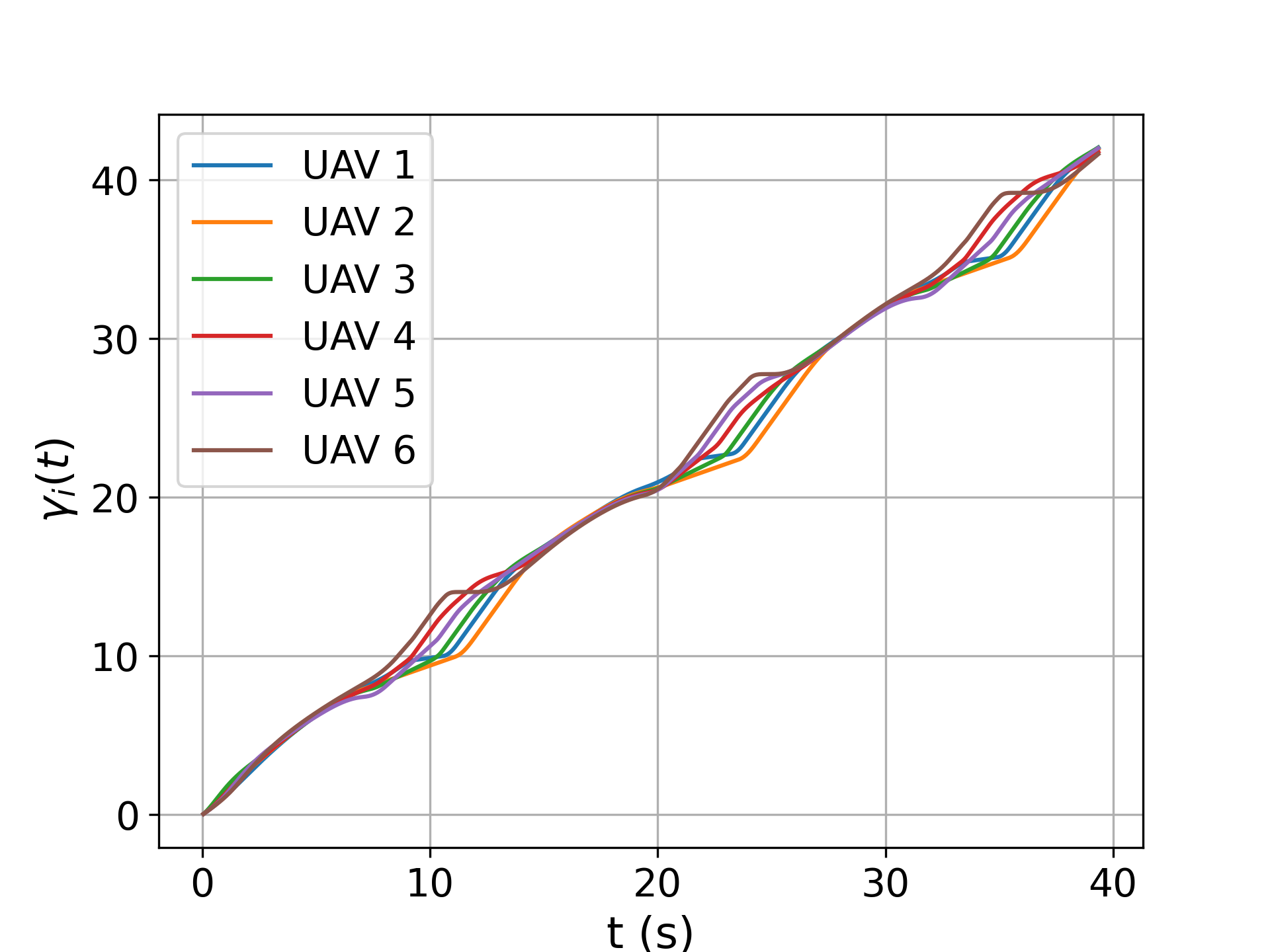}  
        \caption{}
        \label{fig:scen4_gamma}
    \end{subfigure}
    \hfill
    \begin{subfigure}[b]{0.28\textwidth}
        \centering
        \includegraphics[width=\textwidth]{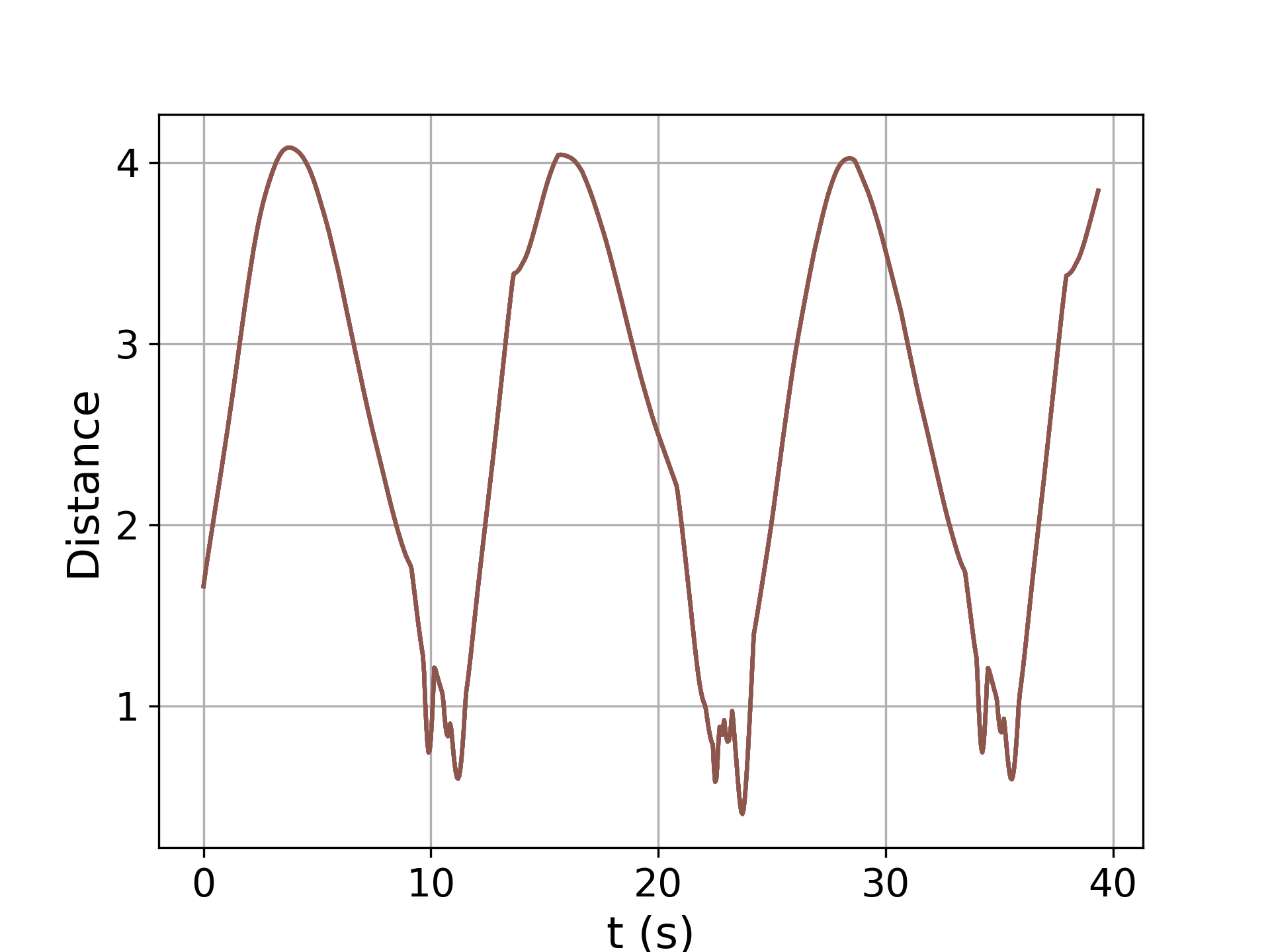}
        \caption{}
        \label{fig:min_distance}
    \end{subfigure}
    \caption{Collision avoidance under non-ideal communication and ideal path following: (a) A top view of the actual trajectories followed by the UAVs, with red circles marking the starting points and stars indicating the final positions; (b) $\gamma_i$ over time; (c) Minimum distance maintained between any two UAVs over time.}
    \label{fig:Scenario4}
\end{figure*}
\subsubsection*{Scalability} The method is inherently distributed and relies only on local information, making it naturally scalable. Simulations confirm practical applicability for systems with several dozen UAVs. To test the scalability of the method, we simulated a scenario in which multiple UAVs fly along 
non-overlapping circular trajectories that share a common center but have different radii. 
We assume ideal communication and path-following conditions in this scenario. 
Under these assumptions, as the number of agents increases, the computation time grows only slightly, 
as shown in Table~\ref{tab:data}. 
This limited growth in computation time is due to the fact that information from other agents enters 
each optimization problem only through fixed coupling terms in the cost function, while the 
dimensionality of the decision variable remains one.

\begin{table}[ht]
\centering
\resizebox{\columnwidth}{!}{ 
\begin{tabular}{|c|c|c|c|}
\hline
\textbf{\# of UAVs} & \textbf{Mean Time}  & \textbf{Max Time} & \textbf{Consensus Time} \\
\hline
10  & 0.005834  & 0.014842 & 4.8 \\
20  & 0.007044  & 0.022113 & 4.2\\
30  & 0.007960 &  0.027910 & 4.1 \\
\hline
\end{tabular}
}
\caption{Mean, Max MPC step calculation times and consensus achievement time in ideal scenario.}
\label{tab:data}
\end{table}
The simulation results show the performance of the proposed coordination algorithm in various scenarios. Under ideal conditions, UAVs achieve synchronization rapidly. Introducing non-ideal communication (including partial communication failures) and wind disturbances reveals the ability of the algorithm to handle real-world challenges, ensuring coordination and maintaining safety through collision avoidance mechanisms. Additionally, the method's inherent scalability, demonstrated through simulations with increasing numbers of UAVs, confirms its practical applicability. For reference, all simulation videos can be found at the following \href{https://www.youtube.com/watch?v=yaI_HdUjJj4&t=77s}{link}. The RotorPy simulation code and additional results are available at \href{https://github.com/amanucha/rotorpy_coordination}{this
 link}, and simulations at the Crazyswarm platform can be found at \href{https://github.com/mikayel2/swarm_timecoord_mpc}{this
 link}.

\section{Conclusion and Future Work}
\label{sec:conclusion}

In this paper, we proposed a novel game-theoretic framework for time-critical cooperative missions of UAV systems operating over time-varying networks. The approach introduces a distributed time-coordination mechanism that ensures agile, scalable, system-wide synchronization in dynamic and uncertain environments. The low dimensionality of the proposed optimization problem enables efficient real-time implementation while accommodating UAV dynamical constraints and mission-specific requirements.
In the idealized setting, we proved the existence and exponential stability of a Nash equilibrium, ensuring synchronized system behavior. To address more realistic scenarios, we developed an MPC-based algorithm capable of handling communication failures and path-following errors. Extensive simulations demonstrated the effectiveness and robustness of the approach.
Although the theoretical analysis has so far been conducted only under ideal conditions, in future work we aim to establish the existence and exponential stability of the Nash equilibrium in more general settings, accounting for time-varying communication networks, potential communication failures, and path-following inaccuracies.
These extensions render the cost function nonconvex, while preserving convexity with respect to higher-order derivatives of the virtual time. Consequently, existence of an equilibrium is expected to remain guaranteed, whereas uniqueness may no longer hold. Nevertheless, based on our simulation results, we postulate that all Nash equilibria exhibit exponential convergence properties.

Furthermore, we intend to analyze the convergence properties of the proposed MPC-based algorithm.
\vspace{-0.42cm}
\section*{Acknowledgments}
The authors thank Professors Ant{\'o}nio Pascoal and Isaac Kaminer for their insightful discussions and valuable advice during this research. The authors also extend their thanks to Astghik Hakobyan and Andy Younes for their help with the simulation work. 
\vspace{-0.42cm}
\section{Appendix: Proofs}\label{appen:A}
In this part, we provide detailed proofs of our main results.

   [Proposition \ref{prop-exp-first}]\begin{proof}
Before proceeding to the proof, we briefly outline the main idea. Each agent solves an individual minimization problem that depends on the virtual times of the other agents. Applying standard perturbation techniques yields the corresponding Euler–Lagrange equations and transversality conditions. Collecting these conditions leads to a coupled system of fourth-order differential equations. Exploiting the fully connected communication structure, we decouple and explicitly solve this system using Ferrari’s method. The resulting explicit solutions, together with the transversality conditions, imply exponential convergence to the coordinated state, while convexity of the cost functions guarantees existence and uniqueness of the Nash equilibrium.

Suppose that  $\bm{\gamma}^*=(\gamma_1^*,\dots,\gamma_N^*)\in \prod_{j=1}^{N}\mathcal{B}^{0,\alpha}_j$ is a solution to  Problem \ref{prob-exp-stab}. Then, by the  definition of Nash equilibrium and the convexity of the integrand of \eqref{def-cost-0} it follows that $\gamma_i^*$ is the unique minimizer of the following optimization problem
       \begin{equation}\label{eq-min-prob-App}
\begin{split}
    I^\alpha_{\bm{\gamma}^*}[\gamma_i^*] &= \min_{\gamma_i \in \mathcal{B}^{0,\alpha}_i}  I^\alpha_{\bm{\gamma}^*}[\gamma_i^*] 
= \min_{\gamma_i \in \mathcal{B}^{0,\alpha}_i} \int_0^\infty e^{-\alpha t} 
( w_1\dot{\gamma}_i^2 \\&+ \tfrac{w_2}{N}\sum (\gamma_j - \gamma_j^*)^2 + w_3\ddot{\gamma}_i^2 ) \, dt.
\end{split}
\end{equation}

To examine $\gamma_i^*$ behavior at infinity, we use the Euler-Lagrange equations of \eqref{eq-min-prob-App}. To derive Euler-Lagrange equations, we consider the following 
 perturbations   
\(u^\varepsilon(t) := \gamma_i^*(t) + \varepsilon v(t)\),
where $v \in H^2((0,\infty)),\quad  v(0)=0,\quad \dot{v}(0)=0$. Because $u^\varepsilon\in\mathcal{B}^{0,\alpha}_i$ and $\gamma_i^*$ is a minimizer to  \eqref{eq-min-prob-App},   the scalar function
$\phi(\varepsilon):=I^\alpha_{\bm{\gamma}^*}[u^\varepsilon]$
has minimum at $\varepsilon=0$. Therefore,
\begin{equation}\label{eq-p-main-EL}
\begin{split}
    \phi^\prime(0)=2 \int_0^\infty e^{-\alpha t} 
(
w_1\dot{\gamma}_i^* \dot{v} 
&+ \tfrac{w_2}{N}\sum  (\gamma_i^* - \gamma_j^*)v 
\\&+w_3\ddot{\gamma}_i^*\ddot{v}  )\dt=0,
\end{split}
    \end{equation}
   for any $v \in H^2((0,\infty)),\quad  v(0)=0,\quad \dot{v}(0)=0$.
    First, taking $v \in H_c^2((0,\infty))$ and applying integration by parts and using fundamental lemma of the calculus of variations
 from \eqref{eq-p-main-EL},  we derive the Euler-Lagrange equation
\begin{equation}\label{eq-5-poly-1-App}
\gamma_i^{(4)*}-2\alpha  \dddot{\gamma}_i^*+(\alpha^2-\tfrac{w_1}{w_3}) \ddot{\gamma}_i^* +  \tfrac{\alpha w_1}{w_3} \dot{\gamma}_i^* +\tfrac{w_2}{w_3N}\sum_{j=1}^{N}  (\gamma_i^* - \gamma_j^*)=0,
\end{equation}
for all $i=1,\dots,N$.
Next, using the Euler--Lagrange equation in~\eqref{eq-5-poly-1-App}, 
and applying integration by parts in~\eqref{eq-p-main-EL}, 
we consider two classes of admissible test functions. 
The first class consists of functions 
$v \in H^2((0,\infty))$ satisfying 
$v(0) = 0$, $\dot{v}(0) = 0$, 
$\lim_{T \to \infty} v(T) \neq 0$, 
and $\lim_{T \to \infty} \dot{v}(T) = 0$. 
The second class consists of test functions 
$v$ such that both $\lim_{T \to \infty} v(T) \neq 0$ and 
$\lim_{T \to \infty} \dot{v}(T) \neq 0$. 
By analyzing these two cases, we deduce the corresponding transversality conditions 
\begin{equation}\label{eq-transver}
\begin{split}
   & \lim_{T\to\infty} e^{-\alpha T}\ddot{\gamma}_i^*(T)=0,\\
  &  \lim_{T\to\infty} e^{-\alpha T}(\tfrac{w_1}{w_3}\dot{\gamma_i}^*(T)+\alpha \ddot{\gamma}_i^*(T) - \dddot{\gamma}_i^*(T))=0.
\end{split}
\end{equation}
To examine  $\gamma_i^*$, we use Euler-Lagrange equation in \eqref{eq-5-poly-1-App} and transversality conditions from \eqref{eq-transver}.

To solve the system of coupled fourth-order differential equations in \eqref{eq-5-poly-1-App}, we note that if we subtract the $k^{\rm th}$   equation from $i^{\rm th}$  one  of the system and denote 
\begin{equation}\label{def-y-App}
y_{ik}(t)=\gamma_i^*-\gamma_k^*,
\end{equation}
then, $y_{ik}$ solves the following  fourth-order linear homogeneous ordinary differential equation
\begin{equation}\label{eq-6-poly-3-App}
y^{(4)}_{ik} - 2\alpha \dddot{y}_{ik} + ( \alpha^2-\tfrac{w_1}{w_3})\ddot{y}_{ik} + \tfrac{\alpha w_1}{w_3} \dot{y}_{ik} + \tfrac{ w_2 }{w_3}  y_{ik} = 0.
\end{equation}
To solve the proceeding differential equation, we consider its characteristic equation
\begin{equation}\label{eq-7-poly-App}
	\lambda^{4} - 2\alpha \lambda^3 + ( \alpha^2-\tfrac{ w_1}{w_3})\lambda^2 + \tfrac{\alpha w_1}{w_3} \lambda + \tfrac{ w_2}{w_3}  = 0.
\end{equation}
We study the fourth-degree polynomial equation by Ferrari's method. By the following change of variable 
\begin{equation*}
\lambda=(x+\tfrac{\alpha}{2})
\end{equation*}
from \eqref{eq-7-poly-App}, we obtain a depressed  quartic, which is actually biquadratic
\begin{equation*}
 x^4 - \left(\tfrac{\alpha^2}{2} +\tfrac{ w_1}{w_3}\right)x^2 + \tfrac{\alpha^4}{16} + \tfrac{\alpha^2}{4}\tfrac{ w_1}{w_3} + \tfrac{ w_2}{w_3}.
\end{equation*}
Solving the proceeding equation, we get
\begin{equation*}
x_{1234} = \pm\sqrt{\tfrac{w_3\alpha^2 + 2w_1 \pm 2\sqrt{w_1^2-4w_2w_3}}{4w_3}},
\end{equation*}
Therefore, if $W=w_1^2-4w_2w_3<0$
\begin{equation}\label{eq-def-la-App}
    \begin{split}
\lambda^-_{12} = \frac{\alpha}{2}\pm\sqrt{\tfrac{w_3\alpha^2 + 2w_1 + 2\sqrt{W}}{4w_3}}=\mu^-_1\pm i\nu_1,\\
\lambda^-_{34} =\tfrac{\alpha}{2} \pm\sqrt{\tfrac{w_3\alpha^2 + 2w_1 - 2\sqrt{W}}{4w_3}}=\mu_2^-\pm i\nu_2,
    \end{split}
\end{equation}
where
\begin{equation}\label{eq-def-mu-App}
	\begin{split}
&\mu^-_{12}= \tfrac{\alpha}{2}\mp\tfrac{1}{2}\sqrt{\sqrt{\left(\tfrac{w_3\alpha^2 + 2w_1}{w_3}\right)^2 - \tfrac{4W}{w^2_3}} + \tfrac{w_3\alpha^2 + 2w_1}{2w_3}},\\
&\nu_1=\nu_2=\tfrac{1}{2}\sqrt{\sqrt{\left(\tfrac{w_3\alpha^2 + 2w_1}{w_3}\right)^2 - \tfrac{4W}{w^2_3}} -\tfrac{w_3\alpha^2 + 2w_1}{2w_3}}.
	\end{split}
\end{equation}
Using the solutions to the characteristic equation in \eqref{eq-7-poly-App}, we derive the general form of the  solutions to \eqref{eq-6-poly-3-App}
\begin{equation}\label{eq-sol-DE-App}
\begin{split}
    y_{ik}(t)&= e^{\mu^-_1 t}(A_{ik}\cos(\nu_1 t)+B_{ik}\sin(\nu_1 t))\\&+e^{\mu^-_2 t}(C_{ik}\cos(\nu_2 t)+D_{ik}\sin(\nu_2 t)).
\end{split}
\end{equation}
    It is important to note that the constants $\mu^-_1,\mu^-_2,\nu_1,\nu_2$ do not depend on $i$ and $k$. Furthermore, 
    \eqref{eq-def-mu-App} implies that for any $\alpha>0$ and $N\geq 2$, we have $\mu^-_1<0$, $\mu^-_2>0$. Using these and the 
     transversality conditions in \eqref{eq-transver}, we deduce that $C_{ik}=D_{ik}=0$ in \eqref{eq-sol-DE-3-App}. On the other hand, recalling \eqref{def-y-App} and $\gamma_i^*\in \mathcal{A}_{i}^{0,\alpha},\quad\gamma_k^*\in\mathcal{A}_{k}^{0,\alpha}$, we get the following boundary values for $y_{ik}$ 
     \begin{equation}
         \label{eq-def-boundary}
         y_{ik}(0)=\gamma_{i}^{0}-\gamma_{k}^{0},\quad \dot{y}_{ik}(0)=0.
     \end{equation}
     Relying on the boundary condition in the previous equation, from \eqref{eq-sol-DE-App} we obtain
     \begin{equation}\label{eq-Aik}
        A_{ik}=\gamma^i_0-\gamma^k_0,\quad
	B_{ik}=-\frac{\mu_{1} A_{ik}}{\nu_{1}}.
     \end{equation}
 Therefore, 
\begin{equation}\label{eq-sol-DE-3-App}
	y_{ik}(t)= e^{\mu^-_1 t}(\gamma^i_0-\gamma^k_0)\Big(\cos(\nu_1 t)-\frac{\mu^-_{1} }{\nu_{1}}\sin(\nu_1 t)\Big).
\end{equation}
    Repeating same the arguments for all $(i,k)$ pairs, we get 
\begin{equation*}
    \sum_{j=1}^{N}  (\gamma_i^* - \gamma_j^*)=e^{\mu^-_1 t}(\cos(\nu_1 t)-\tfrac{\mu^-_{1} }{\nu_{1}}\sin(\nu_1 t))S_i,
\end{equation*}
where  $S_i:=\sum_{j=1}^N (\gamma^i_0 - \gamma^j_0)$.
Now, substituting these    into \eqref{eq-5-poly-1-App}, we obtain a non-homogeneous fourth-order ODE with exponential and trigonometric right-hand side
\begin{equation}\label{eq-5-poly-2-App}
\begin{split}
\gamma_i^{(4)*}&-2\alpha  \dddot{\gamma}_i^*+(\alpha^2-\tfrac{w_1}{w_3}) \ddot{\gamma}_i^* + \tfrac{\alpha w_1}{w_3}  \dot{\gamma}_i^* \\&=-\tfrac{w_2}{w_3}e^{\mu^-_1 t}(\cos(\nu_1 t)-\tfrac{\mu^-_{1} }{\nu_{1}}\sin(\nu_1 t)) S_i
\end{split}
\end{equation}
To obtain the general solution to \eqref{eq-5-poly-2-App},  first, we find a particular solution to it. We search for a particular solution  in the following form
\begin{equation}\label{eq-part-sol}
    \gamma_{i}^{p}(t) = e^{\mu^-_1 t} (C^-_{1i} \cos(\nu_1 t) + C^-_{2i} \sin(\nu_1 t)).
\end{equation}
Equalizing right and left hand side of   \eqref{eq-5-poly-2-App} for $\gamma_{i}^{p}$, we find constants $C^-_{1i}$ and $C^-_{2i}$ in terms of $\mu_1,\nu_1,\gamma^0_i,\gamma^0_{-i}$ 
\begin{equation}\label{def-C}
	\begin{cases}
		 C^-_{1i}	 = - \tfrac{\mu_1w_2(P_i^2+Q_i^2)+Q_i(\nu_1P_i-\mu_1Q_i)}{P_iw_3\mu_1(P_i^2+Q_i^2)}S_i\\
          C^-_{2i}
		= \tfrac{w_2(\nu_1P_i-\mu_1Q_i)}{\mu_1w_3(P_i^2+Q_i^2)}S_i,
	\end{cases}
\end{equation}
where
\begin{equation}\label{eq-notation}
	\begin{split}
	P_i&:=\mu_1^4 - 6\mu_1^2 \nu_1^2 + \nu_1^4 - 2\alpha (\mu_1^3 - 3\mu_1 \nu_1^2) \\&\quad\quad\quad\quad+ (\alpha^2 - \tfrac{w_1}{w_3})(\mu_1^2 - \nu_1^2) + \tfrac{\alpha w_1}{w_3} \mu_1,
	\\Q_i&:= 4 (\mu_1^3 - \mu_1 \nu_1^2) \nu_1 -2\alpha (\mu_1^2 - \nu_1^2) \nu_1 \\&\quad\quad\quad\quad+ 2(\alpha^2-\tfrac{w_1}{w_3}) \mu_1 \nu_1+\tfrac{\alpha w_1}{w_3}\nu_1.
	\end{split}
\end{equation}
Next, we provide the general solution to the homogeneous equation of \eqref{eq-5-poly-2-App}
\begin{equation}\label{eq-5-poly-2-h-App}
\begin{split}
\bar{\gamma}_i^{(4)*}&-2\alpha  \dddot{\bar{\gamma}}_i^*+(\alpha^2-\tfrac{ w_1}{w_3}) \ddot{\bar{\gamma}}_i^* + \tfrac{\alpha w_1}{w_3}  \dot{\bar{\gamma}}_i^*=0.
\end{split}
\end{equation}
Similar to the analysis of  \eqref{eq-6-poly-3-App} for  
 the general  solution of the  homogeneous  equation, we get
\begin{equation}\label{eq-homo-gen}
\bar{\gamma}_i=H^-_{1i}+H^-{2i}e^{\alpha t}+H^-_{3i}e^{\mu_3t }+H^-_{4i}e^{\mu_4 t},
\end{equation}
where
\begin{equation}
    \label{eq-mu-3-4}
\begin{split}
\mu_{3}=\tfrac{\alpha}{2}-\sqrt{\tfrac{\alpha^2}{4} +\tfrac{w_1}{w_3}},\quad \mu_{4}=\tfrac{\alpha}{2}+\sqrt{\tfrac{\alpha^2}{4}+\tfrac{w_1}{w_3}}.
\end{split}
\end{equation}
Combining  the particular solution in \eqref{eq-part-sol} with the general solution to the homogeneous equation in \eqref{eq-homo-gen}, we obtain
\begin{equation}\label{eq-sol-General-App}
\begin{split}
\gamma_i^*(t) = H^-_{1i}+H^-_{2i} e^{\alpha t}+H^-_{3i}e^{\mu_3 t}+H^-_{4i}e^{\mu_4 t}\\+e^{\mu^-_1 t} (C^-_{1i} \cos(\nu_1 t) + C^-_{2i} \sin(\nu_1 t)).
\end{split}
\end{equation}
Because $\mu_3<0$ and $\mu_4>0$ from the  
transversality conditions in \eqref{eq-transver}, we deduce that $H^-_{2i}=H^-_{4i}=0$. On the other hand, $\gamma_i^*$ in \eqref{eq-sol-General-App} should satisfy boundary conditions in $\mathcal{B}^{0}_i$.  Therefore,  
\begin{equation}\label{eq-sol-General-boundary-1}
	H^-_{1i}=\gamma_i^0-H^-_{3i}- C^-_{1i},\quad
		 H^-_{3i} =-\tfrac{\mu^-_1 C^-_{1i}+\nu_1 C^-_{2i}}{\mu_3}.
         \end{equation}
This last step proves \eqref{eq-explicit-solution}. In the case $W = w_1^2 - 4N w_2 w_3 > 0$, 
the solutions to the characteristic equation in~\eqref{eq-7-poly-App} are real. 
By arguing as in the previous case, we obtain
\begin{equation}\label{eq-sol-General-App-2}
\begin{split}
\gamma_i^*(t) = H^+_{1i}+H^+_{3i}e^{\mu_3 t}+C^+_{1i} e^{\mu^+_1 t}  + C^+_{2i} e^{\mu^+_2 t},
\end{split}
\end{equation}
where
\begin{equation}\label{eq-sol-General-boundary-1}
	\begin{split}
	C^+_{1i}= \tfrac{1}{A^+}\tfrac{\mu^+_2}{\mu^+_2-\mu^+_1}S_i,
        \,
        C^+_{2i} =-\tfrac{1}{B^+}\tfrac{\mu^+_1}{\mu^+_2-\mu^+_1}S_i,\\
    H^+_{1i}=\gamma_i^0-H^+_{3i}- C^+_{1i}-C^+_{2i},\,
		 H^+_{3i} =-\tfrac{(\mu^+_1C^+_{1i}+\mu^+_2 C^+_{2i})}{\mu_3}
	\end{split}
\end{equation}
with $A^+=(\mu^+_1)^4-2\alpha (\mu^+_1)^3+(\alpha^2-\frac{w_1}{w_3})(\mu^+_1)^2+\frac{\alpha w_1}{w_3}\mu^+_1$, $B^+=(\mu^+_2)^4-2\alpha (\mu^+_2)^3+(\alpha^2-\frac{w_1}{w_3})(\mu^+_2)^2+\frac{\alpha w_1}{w_3}\mu^+_2$
and 
\begin{equation}\label{eq-def-la-App-2}
    \begin{split}
\mu^+_{12} = \tfrac{\alpha}{2}-\sqrt{\tfrac{w_3\alpha^2 + 2w_1 \pm 2\sqrt{W}}{4w_3}}.
    \end{split}
\end{equation}

In the third case of $W=0$, with the similar arguments,  we get
\begin{equation}\label{eq-sol-General-App-3}
\begin{split}
\gamma_i^*(t) = H^0_{1i}+H^0_{3i}e^{\mu_3 t}+C^0_{1i} e^{\mu^0 t}  + C^0_{2i}t e^{\mu^0 t},
\end{split}
\end{equation}
where
\begin{equation}\label{eq-sol-General-boundary-3}
	\begin{split}
    C^0_{1i}=-\tfrac{C^0_{2i}B^0+S_i}{A^0},
        \,
        C^0_{2i} =\frac{\mu^0S_i}{A^0},\\
H^0_{1i}=\gamma_i^0-H^0_{3i}- C^0_{1i},\,
		 H^0_{3i} =-\tfrac{(\mu^0 C^0_{1i}+ C^0_{2i})}{\mu_3}
	\end{split}
\end{equation}
with $A^0=(\mu^0)^4-2\alpha (\mu^0)^3+(\alpha^2-\frac{w_1}{w_3})(\mu^0)^2+\frac{\alpha w_1}{w_3}\mu^0$, $B^0=4(\mu^0)^3-6\alpha (\mu^0)^2+2(\alpha^2-\frac{w_1}{w_3})\mu^0+\frac{\alpha w_1}{w_3}$ and
\begin{equation}\label{eq-def-la-App-3}
    \begin{split}
\mu^0 = \tfrac{\alpha}{2}-\sqrt{\tfrac{w_3\alpha^2 + 2w_1}{4w_3}}.
    \end{split}
\end{equation}

To prove the existence of a solution to Problem \ref{prob-exp-stab}, we use backwards arguments. In particular, since $\gamma_i^*$ (given by~\eqref{eq-sol-General-App}, 
or~\eqref{eq-sol-General-App-2}, 
or~\eqref{eq-sol-General-App-3}, 
depending on the value of the parameter $W$) solves the Euler-Lagrange equation in \eqref{eq-5-poly-1-App} and the variational problem in \eqref{eq-min-prob-App} is convex, we have that  $\gamma_i^*$ is the minimizer of \eqref{eq-min-prob-App}. Subsequently, $\bm{\gamma}^*$ is a Nash equilibrium of Problem \ref{prob-exp-stab}.
\end{proof}

\vspace{-0.5cm}
\bibliographystyle{IEEEtran}
\bibliography{references}

\end{document}